\documentclass{article}
\usepackage{arxiv}
\usepackage[utf8]{inputenc} 
\usepackage[T1]{fontenc}    
\usepackage{multirow}
\usepackage{url}            
\usepackage{booktabs}       
\usepackage{amsfonts}       
\usepackage{nicefrac}       
\usepackage{orcidlink}
\usepackage{microtype}      
\usepackage{lipsum}
\usepackage{graphicx}
\graphicspath{ {./images/} }
\RequirePackage{amsthm,amsmath,amsfonts,amssymb}
\RequirePackage{algorithm,paralist}
\RequirePackage[noend]{algpseudocode}
\RequirePackage[numbers]{natbib}

\usepackage{adjustbox}      
\usepackage{xcolor}         
\usepackage{colortbl}       
\usepackage{tikz}           
\usepackage{tcolorbox}      
\usepackage{tabularx} 

\def\b1{\mbox{$\mathbf{1}$}}

\newcommand{\Xvec}{\mathbf{X}}

\newcommand{\epsilonvec}{\boldsymbol{\epsilon}}

\newcommand{\yvec}{\mathbf{y}}

\newcommand{\muvec}{\boldsymbol{\mu}}

\newcommand{\betavec}{\boldsymbol{\beta}}

\newcommand{\phivec}{\boldsymbol{\phi}}

\theoremstyle{plain}

\newtheorem{theorem}{Theorem}[section]

\newtheorem{Lemma}{Lemma}
\newtheorem{condition}{Condition}
\theoremstyle{definition}

\theoremstyle{remark}

\title{Bayesian Models for Joint Selection of Features and Auto-Regressive Lags: Theory and Applications in Environmental and Financial Forecasting}

\author{
Alokesh Manna \orcidlink{0009-0007-7958-5273} \\
Department of Statistics \\
University of Connecticut \\
\texttt{alokesh.manna@uconn.edu} \\
\And
Sujit K. Ghosh \orcidlink{0000-0001-8351-408X} \\
Department of Statistics \\
North Carolina State University \\
\texttt{sujit.ghosh@ncsu.edu}
}

\begin{document}
\maketitle
\begin{abstract}
We develop a Bayesian framework for variable selection in linear regression models with autocorrelated errors that accommodate both lagged covariates and autoregressive error structures. This setting arises naturally in many time-series applications where the response is influenced not only by contemporaneous and past values of explanatory variables but also by persistent stochastic shocks. Examples include financial modeling (e.g., logarithmic returns of asset prices), hydrological forecasting (e.g., groundwater dynamics), and meteorological applications (e.g., temperature or precipitation patterns), where capturing temporal dependencies is critical for accurate inference and prediction.
Our methodology employs a hierarchical Bayesian model with spike-and-slab priors to simultaneously select relevant covariates and lagged error terms. To address the computational challenges of high-dimensional model spaces, we propose an efficient two-stage Markov Chain Monte Carlo (MCMC) algorithm that separates the sampling of variable inclusion indicators and model parameters. Theoretical analysis establishes posterior selection consistency under mild regularity conditions, even when the number of candidate predictors grows exponentially with the sample size, a common feature in modern time-series applications involving many potential lagged variables.
Through extensive simulation studies and two real-data applications, the prediction of groundwater depth and the modeling of S\&P 500 index log-returns—we demonstrate that the proposed approach offers substantial gains in variable selection accuracy and predictive performance. Compared to existing methods, our framework achieves a lower mean squared prediction error (MSPE), improved identification of true model components, and greater robustness in the presence of autocorrelated noise. These results underscore the practical utility of our Bayesian procedure for both model interpretation and forecasting in autoregressive settings.

\end{abstract}


\section{Introduction}  
Vector Autoregression with Exogenous Variables (VARX) is a prevalent approach for forecasting in multivariate time series analysis, leveraging multiple lags of both endogenous and exogenous variables to capture intricate temporal dependencies and interactions. This structure enables a more precise prediction of current values by accounting for the influence of past values of both types of variables (see \cite{ocampo2012introductory}, \cite{nicholson2017bigvar}). Due to the high dimensionality inherent in VARX models, effective variable selection is essential to enhance model interpretability and mitigate overfitting. \cite{dawn2025some}, \cite{dawn2021some} discussed change point detection in a very high-dimensional low number of time points data to account for non-stationarity. In many real-world data scenarios, the time points may be large with low-dimensional, relevant features that can change stochastically. Regularization techniques, notably the lasso introduced by \cite{tibshirani1996regression}, are commonly employed to enforce sparsity by “shrinking” irrelevant coefficients towards zero, thus improving computational efficiency and preserving only the most influential variables. The ``glmnet'' package by \cite{friedman2021package} provides a computational tool that integrates sparsity within a linear regression framework. Building on this, \cite{nicholson2017bigvar} proposed structured penalties to facilitate the simultaneous estimation of high-dimensional time series within both VAR and VARX frameworks.

In the Bayesian framework, variable selection is typically performed through sparse priors on the regression coefficients, which are designed to promote sparsity. The spike-and-slab prior, as discussed by \cite{ishwaran2005spike}, combines a point mass at zero (the “spike”) with a continuous distribution (the “slab”) over the remaining coefficient values, enabling selective shrinkage. Alternatively, the horseshoe prior, explored by \cite{carvalho2009handling}, applies global and local shrinkage parameters to the coefficients, providing a flexible approach for managing sparsity while accommodating signals of varying strengths across the coefficients. These Bayesian sparsity-inducing priors offer robust alternatives for variable selection in high-dimensional VARX settings, allowing for improved model accuracy and interpretability.
\cite{goh2018bayesian} investigated variable selection using Gaussian and diffused-gamma prior with $l_{0}$ norm penalization. \cite{manna2024hydrology} explored water table depth prediction using the ``BigVAR'' package using a frequentist elastic net penalized regression method for variable selection. \cite{manna2024interval} explored interval estimation of penalized regression in a complicated Tweedie model framework in the Generalized Linear Model (GLM). Conformal inference for quantifying uncertainty in prediction intervals using GLM-type residuals is also developed in \cite{manna2025distributionfreeinferencelightgbmglm}. While the above are frequentist methods, this paper presents a Bayesian approach that effectively quantifies uncertainty for non-IID data, such as time series, while achieving high predictive accuracy.

While variable selection and prediction within an autoregressive framework, especially the simultaneous selection of coefficients for both endogenous and exogenous variables and their lags have been explored, there remains a gap in the Bayesian framework. To address this the strong posterior consistency of the coefficients for both exogenous and endogenous variables within an autoregressive model was investigated, providing a Bayesian approach to variable selection and prediction.
\cite{chen2019fast} considered the model
$
\label{chenwalker_model}
\mathbf{y} = \mathbf{X}\betavec + \epsilonvec, 
$
where \(\mathbf{y}\) is a \(n \times 1\) vector of response variables, \(\mathbf{X}\) is a \(n \times p\) matrix of predictor variables, \(\betavec\) is a \(p \times 1\) vector of coefficients, \(\sigma^2\) is an unknown variance term, and \(\epsilonvec\) is a \(n \times 1\) vector of zero-mean errors that are potentially auto-correlated. The columns of the \(\mathbf{X}\) may contain the lagged values of the feature or predictor variables. Thus, the total predictors $p$ can be larger than the time points $n$. In this work, it is assumed that these errors follow an AR(q) model described more briefly as follows:
$
\label{ar_error}
\epsilon_t = \sum_{l=1}^q\phi_l\epsilon_{t-l} + u_t, \quad |\phi_j| < 1, \quad j = 1, 2, \dots, q;
$
where \(u_t\) are assumed to be independently identically distributed (i.i.d.) errors satisfying \(E(u_t) = 0\), \(\text{Var}(u_t) = \sigma^2\), and they are independent of the predictor variables (see \cite{tuacc2021variable}). A similar model with auto-regressive error has also been explored in \cite{wang2007regression} where $l_{1}$ penalty was introduced over the coefficients for variable selection. Thus, the model incorporating lag structures for both errors and covariates is formulated as follows:
\begin{align}
\label{proposed_model_formula}
&y_t=\mathbf{\nu}+\sum_{j=1}^r \mathbf{x}_{t-j}^\prime\betavec_{j}+\epsilon_t\; \text{ for } \;t=1,\ldots,n,\\
& \epsilon_t = \sum_{l=1}^q\phi_l\epsilon_{t-l} + u_t, \quad |\phi_j| < 1, \quad j = 1, 2, \dots, q.
\label{eq:two_stage_model}
\end{align}
In this model, predictors are represented as $\{\mathbf{x}_t\}_{t=1}^{n}$ and the dependent variable is represented as $\{ y_t\}_{t = 1}^n $. $\betavec_{j}$ is the coefficient vector of $\mathbf{x}_{t-j}$, $j^{th}$ lag of the predictors and $\phi_{\ell}$ is the $l^{th}$ coefficient for the lag structure of $y_{t}$. We use ``$r$'' for the number of lags used for the covariates. An additive random Gaussian noise $u_t\stackrel{\text{wn}}{\sim}(0,\sigma_{u}^{2})$ is considered in the model. For This model, the variable selection was examined incorporating elastic net penalty over the coefficients $[\mathbf{\Phi},\betavec]$ in \cite{nicholson2017bigvar}. Define $\mathbf{\Phi}:=\left(\phi_{1},\phi_{2},\cdots,\phi_{q}\right)$ as the vector of all coefficients of the lag structure of $y_{t}$ with order ``q'' and $\betavec:=\left(\betavec_{1},\betavec_{2},\cdots,\betavec_{p}\right)$ be the vector of coefficients of the predictors with their lag order ``p''. \cite{ghosh2021strong} proposed model is based on a pseudo-likelihood-based Bayesian approach for
variable selection, and then examined model selection consistency for the VAR model under the assumption of sparsity in $\mathbf{\Phi}$ for the model $\mathbf{y}=\Xvec\mathbf{\Phi}+\epsilonvec$, where the response vector $\mathbf{y}$ is $n \times 1$, and $\boldsymbol{\varepsilon} \sim \mathcal{N}_n(0, \mathbf{\Sigma})$. In the pseudo-likelihood approach, $\mathbf{\Sigma}$ was assumed to be diagonal, and  $\mathbf{\Phi}$ was estimated using a spike and slab prior. One of the primary goals of this work is the {\em simultaneous selection} of the number of feature variables $p$ and the number of the lags $q$ of the AR errors process. It may happen that $p+q>n$ but the true model is assumed to be sparse in the sense that most of the coordinates of $\betavec$ and $\mathbf{\Phi}$ are zeros. The above-described models assume that the lag ``p'' and ``q'' are fixed. In the model presented in section \ref{proposed_model}, sufficiently large values for both ``p'' and ``q'' were considered. At each step, the simultaneous selection of different lag values for the independent and dependent variables was allowed.

In most of the Bayesian literature, Markov Chain Monte Carlo (MCMC) methods are employed to explore the entire model space of dimension \( 2^{q+pr} \), where \( r \) represents the number of available covariates. However, when \( q + pr \) is large, this can become computationally expensive. To address this issue, \cite{chen2019fast} proposed an analytically tractable solution by utilizing a spike-and-slab prior. The model described in \ref{chenwalker_model} can be rewritten as
\[
\mathbf{y} = x_j \beta_j + \mathbf{X}_{[-j]} \betavec_{[-j]} + \sigma \epsilon
\]
for \(j = 1, \dots, p\). Here, \(x_j\) is the \(j\)-th column of the design matrix \(\mathbf{X}\), and \(\mathbf{X}_{[-j]}\) is the design matrix \(X\) without the \(j\)-th column. With suitable normal priors for \((\beta_j , \betavec_{[-j]})\) and \(\sigma^2\), \cite{chen2019fast} obtained a closed-form expression for the posterior distribution of \((\beta_j , \betavec_{[-j]})\). In our paper, the above idea was extended for the model \ref{proposed_model_formula} and obtained the conditions for posterior consistency. This paper established strong selection consistency of the proposed method, i.e., the posterior probability assigned to the true underlying VAR model converges when the dimension of the VAR system grows nearly exponentially with the number of time points.\\
The rest of the paper is organized as follows. Firstly the description of the newly proposed model is in section \ref{proposed_model} which incorporates autocorrelated errors with exogenous variables with lags. After that, section \ref{theory} explores all the conditions and theorems we derive for posterior selection consistency for the variables with spike and slab prior. This incorporates not only selection of the exogenous variables but also the coefficients for the persistent shocks in autocorrelated errors applicable to even exponential size of the number of predictors and lags of the error process with respect to number of samples. Then the novel two-stage estimation procedure for variable selection and prediction was discussed in section \ref{estimation_procedure}. In section \ref{simulation_results}, we explain true positive, false positive, true negative, false negative, and accuracy for a synthetic dataset with high prediction accuracy. After providing motivation for the proposed model, we explain our results for two real datasets: a) water table depth prediction in the presence of several hydro-climatic variables, b) S\&P 500 stock market data prediction with the other available stocks in section \ref{real_data_application} with a longer historical dataset from 2007 and another based on a recent data from 2021. Finally, we keep our discussion and conclusion in section \ref{Discussion} with acknowledgment in section \ref{Acknowledgement}. We have presented the proof of the results in the section \ref{appendix}. For a more detailed description and analysis of the real data with the performance metrics, please consider \ref{appendix}.

\section{Bayesian Model for Joint Variable and Lag Selection}
\label{proposed_model}
A Bayesian hierarchical model is a statistical model that organizes parameters and lagged variables into multiple levels, where each level represents a different layer of variability or dependency structure. These models are particularly useful when data have a natural hierarchical or nested structure such as in time series models. The two-stage model described in eq. (\ref{proposed_model_formula}) and (\ref{eq:two_stage_model}) can be combined to be written as
\begin{equation}
\label{rewritten_proposed_model}
y_t \mid y_s, s<t \sim N\left(\mu_t +\sum_{l=1}^q\phi_l\left(y_{t-l}-\mu_{t-l}\right), \sigma^{2}_{\epsilon}\right).\; \
\end{equation}
where $\mu_t = \nu + \sum_{j=1}^p\mathbf{x}_{t-j}^\prime\betavec_{j}$. Notice that the above representation ensures that marginally $E(y_t) = \mu_t$ irrespective of the lag structure assumed for the AR(q) errors. However, the conditional mean is given by $E(y_t | y_s, s<t) = \mu_t +\sum_{l=1}^q\phi_l\left(y_{t-l}-\mu_{t-l}\right) $.
We can concatenate all of the $\betavec_j$'s into a big vector $\betavec^\prime = (\betavec_1^\prime,\betavec_2^\prime,\ldots,\betavec_p^\prime)$. Spike-and-slab priors are widely used in Bayesian linear models to enable variable selection, especially in high-dimensional settings where selecting relevant predictors is crucial. This prior distribution promotes sparsity by distinguishing between important and unimportant coefficients, allowing only a subset of variables to influence the model while significantly setting others close to zero. The prior support should be large enough as entire real values, which contain the original values of the coefficients for $\betavec$ and $\phivec$. The relationship between $\yvec$ and $\epsilonvec$ in model \ref{proposed_model_formula} is highly non-linear and discussed in section \ref{theory}.  We utilize the spike and slab prior specifications for the $\betavec$ and $\phivec$ parameters of the above linear model with auto-correlated errors as presented in (\ref{rewritten_proposed_model}) as follows:
\begin{equation}
 \label{prior_specification_beta}
\beta_j \mid Z^{\beta}_{j} , \sigma^2 \sim 
\begin{cases}
N(0, \sigma^2(\tau^{\beta}_{0n})^2) & \text{if } Z^{\beta}_{j} = 0, \\
N(0, \sigma^2(\tau^{\beta}_{1n})^2) & \text{if } Z^{\beta}_{j} = 1,
\end{cases}
\end{equation}
for $j=1,2,\ldots,pr$, and
\begin{equation}
 \label{prior_specification_phi}
\phi_l \mid Z^{\phi}_{l} , \sigma^2 \sim 
\begin{cases}
N(0, \sigma^2(\tau^{\phi}_{0n})^2) & \text{if } Z^{\phi}_{l} = 0, \\
N(0, \sigma^2(\tau^{\phi}_{1n})^2) & \text{if } Z^{\phi}_{l} = 1,
\end{cases}
\end{equation}
for $l=1,2,\ldots,q$. The indicator random variable $Z^{\beta}_{j}$ and $Z^{\phi}_{l}$ with the selection probability are defined as follows:
\begin{equation}
\label{selection_probability}
    P(Z^{\beta}_{j} = 1) = 1 - P(Z^{\beta}_{j} = 0) = q^{\beta}_{jn},\quad
    P(Z^{\phi}_{l} = 1) = 1 - P(Z^{\phi}_{l} = 0) = q^{\phi}_{ln}.
\end{equation}
The indicator variables $Z^{\beta}_{j}$ and $Z^{\phi}_{l}$ are introduced to determine if the coefficients are active or inactive. If the coefficient is active, the prior is considered to be a flat prior whereas, for the inactive coefficients, the prior is considered as spike or heavily concentrated at zero having a Gaussian distribution. We consider the prior for the rest of the coefficients $\beta_{-j}$ and $\phi_{-l}$ conditioned on $\sigma^2$ as
\begin{equation}
\label{prior_to_other_coeffficients}
\beta_{-j} \mid \sigma^2 \sim N(0, \sigma^2\tau^2_{n}I), \quad \phi_{-l} \mid \sigma^2 \sim N(0, \sigma^2\tau^2_{n}I),\quad \sigma^2 \sim IG(a, b),  
\end{equation}
where $IG$ stands for Inverse Gamma distribution. Here we assume \( \mathbf{y} \) and each column of \( \mathbf{X} \) are standardized such that
$
\sum_{i=1}^n x_{ij} = 0, \quad \sum_{i=1}^n x_{ij}^2 = n, \quad\sum_{i=1}^n y_i = 0, \quad \sum_{i=1}^n y_i^2 = n.
$
For theoretical exploration, we assume that the prior probability for being an active variable the same $q^{\beta}_{jn}=q^{\beta}_{n}$ and $q^{\phi}_{jn}=q^{\phi}_{n}$. The primary goal of the posterior inference is to identify the true subset of $\betavec$ coefficients that are non-zero and the true subset of $\phivec$ coefficients that are non-zero. In other words, we would like to identify the true subset $S_\beta \subset \{1,2,\ldots,pr\}$ where $S_\beta = \{j: \beta_j\neq 0\}$ and simultaneously the true subset $S_\phi \subset \{1,2,\ldots,q\}$ where $S_\phi = \{j: \phi_j\neq 0\}$ and $q<n$. By selecting only the most relevant variables and reducing overfitting, spike-and-slab priors often yield more accurate predictions on out-of-sample data, which is instrumental in applications like finance, genomics, and environmental modeling. Next, we present, some asymptotic justifications of the selection consistency of the proposed model under suitable regularity conditions.

\section{Posterior Consistency}
\label{theory}
Posterior consistency in Bayesian variable selection refers to the asymptotic concentration of the posterior distribution around the true model as the sample size increases. In the context of autoregressive models with high-dimensional covariates, achieving posterior consistency is particularly challenging when one simultaneously selects among current and lagged covariates, as well as among lagged response terms that characterize the autocorrelation structure.

The primary challenge stems from the combinatorial complexity of the model space. When each predictor is considered along with multiple lags, the number of candidate variables increases rapidly, often exponentially with the number of time points. Furthermore, the inclusion of lagged responses introduces strong temporal dependencies that can obscure the marginal contribution of individual predictors, especially in the presence of multicollinearity or when lag effects are subtle or delayed. These issues complicate both the computational and inferential aspects of variable selection, making it difficult to disentangle the true signal from spurious correlations.

From a theoretical point of view, establishing posterior consistency in this setting requires careful control of the prior concentration and model complexity. Spike-and-slab priors, which place a point mass at zero (the “spike”) and a diffuse prior over nonzero values (the “slab”), provide a flexible mechanism for sparse modeling. However, ensuring that the posterior correctly identifies the nonzero coefficients in the presence of autocorrelated errors requires stronger regularity conditions. These include assumptions on the design matrix (e.g., eigenvalue bounds), sparsity level, prior hyperparameters, and tail behavior of the error distribution. Moreover, the dependence structure of the response process introduces additional layers of complexity, necessitating results from dependent data theory to ensure that the posterior remains well behaved as the number of time points increases.

Despite these challenges, the simultaneous selection of predictors, their lagged values, and autoregressive components offers substantial advantages in terms of modeling fidelity and forecasting power. By capturing both contemporaneous and delayed effects of exogenous variables, as well as the endogenous dynamics of the response variable, the resulting models are better suited for complex time series data encountered in practice. Importantly, when posterior consistency is established under such a framework, it guarantees that the Bayesian procedure not only recovers the true subset of influential variables, but also accurately quantifies uncertainty, leading to improved interpretability and generalization.


To take advantage of the existing methodologies to establish posterior consistency of the linear model with independent errors, we utilize a transformation that brings our model to a classical linear model with independent errors. With a slight abuse of notations, we concatenate the lagged values of the predictors $\boldsymbol{x}_{t}$'s to a big $p_{0}r\times 1$ vector $\boldsymbol{x}_t^\prime \leftarrow (\boldsymbol{x}_{t}^\prime,\boldsymbol{x}_{t-1}^\prime, \boldsymbol{x}_{t-2}^\prime,\ldots,\boldsymbol{x}_{t-r+1}^\prime)$; i.e. we include ``r-1'' many lags of each ``$p_{0}$'' covariates. We use ``p'' for simplified notation as the total number of covariates. To illustrate the concept under a simple setting when $q=1$ consider first an AR(1) process $y_t$ with covariates $\boldsymbol{x_t}$ as:
\[
y_t = \mu_t + \phi(y_{t-1} - \mu_{t-1}) + \epsilon_t, \quad t = 2, \ldots, n
\]
with initial condition
\[
y_1 = \mu_1 + \epsilon_1,
\]
 where $\mu_t=\boldsymbol{x_t}^{'}\boldsymbol{\beta}$. Assume that $\epsilon_t \stackrel{iid}{\sim} N(0, \sigma^2)$.
Hence, this model can be written as 
\[\tilde{y}_t = \tilde{\boldsymbol{x}}_t^{'}\boldsymbol{\beta} +\epsilon_t,\] 
where $\tilde{y}_t = y_t -\phi y_{t-1}$ and $\tilde{\boldsymbol{x}}_t = \boldsymbol{x}_t - \phi\boldsymbol{x}_{t-1}$ for $t=2,\ldots,n$. In general, the model can be written for any general order as
\begin{equation}
    \label{final_model}
\mathbf{A}(\phivec) \mathbf{y}= \mathbf{A}(\phivec)\muvec +\epsilonvec,
\end{equation}
where $\mathbf{y}=\left(y_{1},y_{2},\cdots,y_{n}\right)^{T}$, $\muvec:=\left(\mu_{1},\mu_{2},\cdots,\mu_{n}\right)^{T}$ and $\epsilonvec:=\left(\epsilon_{1},\epsilon_{2},\cdots,\epsilon_{n}\right)^{T}$. Here we assume that $\epsilon_i \overset{\text{i.i.d.}}{\sim} \mathcal{N}(0, \sigma^2)$, $\mu_{i}=\boldsymbol{x_i}^t\boldsymbol{\beta}$ and $\mathbf{A}(\phivec)$ is a matrix, where $\phivec=\left(\phi_{1},\phi_{2},\cdots,\phi_{q}\right)$.  For AR(1), \[\mathbf{A}(\phivec) =
\begin{pmatrix}
1 & 0 & 0 & \cdots & 0 \\
-\phi & 1 & 0 & \cdots & 0 \\
0 & -\phi & 1 & \cdots & 0 \\
\vdots & \vdots & \vdots & \ddots & \vdots \\
0 & 0 & 0 & \cdots & 1
\end{pmatrix}.\] For any general autoregressive model of order \(q>1\), the matrix \(\mathbf{A}\) can be defined as equation \ref{mat_A}. We denote the elements of the matrix as \(A_{ij}\). For a general order \(p\), the matrix elements \(A_{ij}\) are given by:

\[
\label{mat_A}
A_{ij} = 
\begin{cases} 
1 & \text{if } i = j, \\
-\phi_{i - j} & \text{if } 1 \leq i - j \leq q, \\
0 & \text{otherwise}.
\end{cases}
\]

Thus, the matrix \( \mathbf{A}(\phivec) \) for general order \( q \) is:

\[
\mathbf{A}(\phivec) =
\begin{pmatrix}
1 & 0 & 0 & \cdots & 0 & 0 & \cdots & 0 \\
-\phi_1 & 1 & 0 & \cdots & 0 & 0 & \cdots & 0 \\
-\phi_2 & -\phi_1 & 1 & \cdots & 0 & 0 & \cdots & 0 \\
\vdots & \vdots & \vdots & \ddots & \vdots & \vdots & \cdots & \vdots \\
-\phi_q & -\phi_{q-1} & -\phi_{q-2} & \cdots & -\phi_1 & 1 & \cdots & 0 \\
0 & -\phi_q & -\phi_{q-1} & \cdots & -\phi_2 & -\phi_1 & \cdots & 0 \\
\vdots & \vdots & \vdots & \ddots & \vdots & \vdots & \ddots & \vdots \\
0 & 0 & 0 & \cdots & -\phi_q & -\phi_{q-1} & \cdots & 1 \\
\end{pmatrix}.
\]
Here, each element \(A_{ij}\) of the matrix \(\mathbf{A}(\phivec)\) is defined according to the conditions above, with \(1\) on the main diagonal (\(i = j\)), \(-\phi_{i - j}\) in the first \(p\) sub-diagonals below the main diagonal, and \(0\) elsewhere. Thus one can write the model taking inverse of $\mathbf{A}(\phivec) $ as\[
\mathbf{y}= \muvec + (\mathbf{A}(\phivec) )^{-1}\epsilonvec,
\] where \[\muvec=\mathbf{X}\betavec=\mathbf{x}_{j}\beta_{j}+\mathbf{X}_{[-j]}\betavec_{[-j]},\text{ for any column } j\in\{1,2,\cdots,k=pr\}.\]  
The relationship with $\yvec$ and $\epsilonvec$ is very non-linear and allows for a flexible range of patterns of dependence among values of the process \(\yvec\), expressed in terms of the independent (unobservable) random shocks \(\epsilonvec\). To obtain this process stable, stationarity is a crucial assumption in time series analysis and econometrics, particularly for models like autoregressive (AR), moving average (MA), and more generally for ARMA processes. For stationarity, one requires that the roots of the characteristic polynomial lie outside the of the interval $[-1, 1]$. In particular, if the AR(q) process is represented as $\Phi(B)\epsilon_t = \mathbf{u}_t$, where $\Phi(B)=1-\sum_{l=1}^q\phi_l B^l$ and $\omega_t$ is a white noise process, then we can represent $\epsilon_t = \Phi(B)^{-1}\mathbf{u}_t=\sum_{j=0}^\infty\psi_j B^j\mathbf{u}_t$, the stationarity of the AR process can be expressed by the following abstract condition:
\begin{condition}
\label{cond:0}
The series $\Phi(B)^{-1}=1 + \psi_1 B + \psi_2 B^2 + \cdots$ is absolutely square summable in the sense that \( \sum_{j=0}^{\infty} \psi_j^2 < \infty \), (see \cite{box2015time}, section 1.2.1.)
\end{condition}

\noindent The conditional likelihood function (for a given $\phivec$) for the model \ref{final_model} can be expressed as:
\begin{align*}
    &f(\mathbf{y}|\betavec, \sigma^2,\phivec)= \frac{1}{(2 \pi \sigma^2)^{n/2} } \exp \left( -\frac{1}{2 \sigma^2} (\mathbf{A}(\phivec)\mathbf{y} - \mathbf{A}(\phivec)\mathbf{X} \betavec)^T (\mathbf{A}(\phivec)\mathbf{y} - \mathbf{A}(\phivec)\mathbf{X} \betavec) \right)\\
    &=\frac{1}{(2 \pi \sigma^2)^{n/2} } \exp \left( -\frac{1}{2 \sigma^2} (Y - \mathbf{X} \betavec)^T (\mathbf{A}(\phivec)^{T}\mathbf{A}(\phivec))(Y - \mathbf{X} \betavec) \right)\\
    &=\frac{1}{(2 \pi \sigma^2)^{n/2} } \exp \left( -\frac{1}{2 \sigma^2} (Y -\mathbf{x}_{j}\beta_{j}-\mathbf{X}_{[-j]}\betavec_{[-j]})^T (\mathbf{A}(\phivec)^{T}\mathbf{A}(\phivec))(Y - \mathbf{x}_{j}\beta_{j}-\mathbf{X}_{[-j]}\betavec_{[-j]}) \right).
\end{align*}
To simply notations, we are going to drop the conditioning parameter $\phivec$ as the derivations are conditioned on given arbitray value of $\phivec$. The objective is to calculate $p(\mathbf{y}|\beta_j, \sigma^2,\tau^2_n )$ to understand the likelihood behavior depending on each of the coefficients $\beta_j$. Observe that,
\[
\mathbf{y}= \mathbf{x}_{j}\beta_{j}+\mathbf{X}_{[-j]}\betavec_{[-j]} + (\mathbf{A}(\phivec) )^{-1}\epsilonvec
\overset{\text{i.i.d.}}{\sim} \mathcal{N}(\mathbf{0},\sigma^2\left(\tau^2_n \mathbf{X}_{[-j]}\mathbf{X}_{[-j]}^{T}+(\mathbf{A}(\phivec)^{T}\mathbf{A}(\phivec))^{-1}\right)).
\]
If we multiply by \( p(\betavec_{[-j]}|\sigma^2) \) which is \( \mathcal{N}(0, \sigma^2 \tau^2_n I) \), and integrate over \( \beta_{[-j]} \), we obtain
\begin{align*}
&p(\mathbf{y}|\beta_j,\phivec, \sigma^2) \\
&\propto \sigma^{-n} |\tau^2_n \mathbf{X}_{[-j]}\mathbf{X}_{[-j]}^{T}+(\mathbf{A}(\phivec)^{T}\mathbf{A}(\phivec))^{-1}|^{-1/2} \exp \left( -\frac{1}{2\sigma^2} (\mathbf{y}-\mathbf{x}_{j}\beta_{j}
)^T (I - \mathbf{\tilde{H}}_j)(\mathbf{y}-\mathbf{x}_{j}\beta_{j}
) \right),
\end{align*}
where $(I - \mathbf{\tilde{H}}_j)=\left(\tau^2_n \mathbf{X}_{[-j]}\mathbf{X}_{[-j]}^{T}+(\mathbf{A}(\phivec)^{T}\mathbf{A}(\phivec))^{-1}\right)^{-1}$.
Under the conditions of Lemma \ref{lem:psd_matrix}, $ (I - \mathbf{\tilde{H}}_j)$ is positive semi-definite. So, the likelihood is maximized when 
\begin{equation}
\label{beta_j_hat}
    \hat{\beta}_j = \frac{\mathbf{x}_{j}^T (I - \mathbf{\tilde{H}}_j) Y}{x_j^T (I - \mathbf{\tilde{H}}_j) \mathbf{x}_j}.
\end{equation}
Our goal is to evaluate the distribution $p(\beta_j \mid \mathbf{y},\phivec )$. We can observe that,
\[
p(\beta_j \mid \mathbf{y}, \phivec )\propto
\int_{0}^{\infty} p(\mathbf{y} \mid \beta_j , \sigma^2)p(\beta_j \mid \sigma^2) d\sigma^2. 
\]
We can write $p(\beta_j \mid \mathbf{y} )$ as,
\[
p(\mathbf{y} \mid \beta_j , \sigma^2)p(\beta_j \mid \sigma^2) \propto (1-q^{\beta}_{jn}) \omega_{0j} N(\beta_j \mid \mu_{0j} , \xi^2_{0j}) + q^{\beta}_{jn} \omega_{1j} N(\beta_j \mid \mu_{1j} , \xi^2_{1j}),
\]
where for $k=0, 1$;
\begin{align*}
\omega_{kj} &= \sqrt{\frac{(\tau^{\beta}_{kn})^{-2}}{\mathbf{x}_{j}^T (I - \mathbf{\tilde{H}}_j) \mathbf{x}_{j} + (\tau^{\beta}_{kn})^{-2}}} \exp \left( \frac{1}{2\sigma^2} \frac{(\mathbf{x}_{j}^T (I - \mathbf{\tilde{H}}_j) \mathbf{y} )^2}{\mathbf{x}_{j}^T (I - \mathbf{\tilde{H}}_j) \mathbf{x}_{j} + (\tau^{\beta}_{kn})^{-2}} \right),  \\
\mu_{kj} &= \frac{\mathbf{x}_{j}^T (I - \mathbf{\tilde{H}}_j) \mathbf{y}}{\mathbf{x}_{j}^T (I - \mathbf{\tilde{H}}_j) \mathbf{x}_{j} + (\tau^{\beta}_{kn})^{-2}} , \\
\xi^2_{kj} &= \frac{\sigma^2}{\mathbf{x}_{j}^T (I - \mathbf{\tilde{H}}_j) \mathbf{x}_{j} + (\tau^{\beta}_{kn})^{-2}}.
\end{align*}
Given the conjugacy of inverse gamma prior of \(\sigma^2\), we can multiply by \(\text{IG}(\sigma^2 \mid a, b)\) and integrate out \(\sigma^2\), so we obtain a mixture of two non-standardized Student-t distributions. So after integrating out $\sigma^2$,
\[
p(\beta_j \mid \mathbf{y} ) \propto (1 - q^{\beta}_{jn})F_{0j} \, t_0(\beta_j \mid \nu, \mu_{0j}, \psi_{0j}) + q^{\beta}_{jn}F_{1j} \, t_1(\beta_j \mid \nu, \mu_{1j}, \psi_{1j})
\]
where $\nu = n + 2a$, and for $ k \in \{0, 1\}$,
\begin{align*}
\psi_{kj} &= \frac{b + \frac{1}{2} \mathbf{y}^T (I - \mathbf{\tilde{H}}_j) \mathbf{y} - \frac{1}{2} \frac{(\mathbf{x}_{j}^T (I - \mathbf{\tilde{H}}_j) \mathbf{y} )^2}{\mathbf{x}_{j}^T (I - \mathbf{\tilde{H}}_j) \mathbf{x}_{j} + (\tau^{\beta}_{kn})^{-2}}}{(n + 2a)(\mathbf{x}_{j}^T (I - \mathbf{\tilde{H}}_j) \mathbf{x}_{j} + (\tau^{\beta}_{kn})^{-2})},  \\
F_{kj} &=
\sqrt{\frac{(\tau^{\beta}_{kn})^{-2}}{\mathbf{x}_{j}^T (I - \mathbf{\tilde{H}}_j) \mathbf{x}_{j} + (\tau^{\beta}_{kn})^{-2}}}\times \\
&\left( b + \frac{\mathbf{y}^T (I - \mathbf{\tilde{H}}_j) \mathbf{y}}{2} - \frac{(\mathbf{x}_{j}^T (I - \mathbf{\tilde{H}}_j) \mathbf{y} )^2}{2 (\mathbf{x}_{j}^T (I - \mathbf{\tilde{H}}_j) \mathbf{x}_{j} + (\tau^{\beta}_{kn})^{-2})} \right)^{-\left( \frac{n}{2} + a \right)}.
\end{align*}
Now observe that,
\[\label{mu}
\mu_{kj} \sim \mathcal{N}\left(\frac{\mathbf{x}_{j}^T (I - \mathbf{\tilde{H}}_j) \mathbf{X} \betavec^{*} }{\mathbf{x}_{j}^T (I - \mathbf{\tilde{H}}_j) \mathbf{x}_{j} + (\tau^{\beta}_{kn})^{-2}} , \frac{\mathbf{x}_{j}^T (I - \mathbf{\tilde{H}}_j)   (\mathbf{A}(\phivec)^{T}\mathbf{A}(\phivec))^{-1}
(I - \mathbf{\tilde{H}}_j)\mathbf{x}_{j}
}{ (\mathbf{x}_{j}^T (I - \mathbf{\tilde{H}}_j) \mathbf{x}_{j} + (\tau^{\beta}_{kn})^{-2})^{2} } \right),
\]
where 
 \begin{align}
     &\frac{\mathbf{x}_{j}^T (I - \mathbf{\tilde{H}}_j) \mathbf{X} \betavec^{*} }{\mathbf{x}_{j}^T (I - \mathbf{\tilde{H}}_j) \mathbf{x}_{j} + (\tau^{\beta}_{kn})^{-2}}\\
     &=\frac{\mathbf{x}_{j}^T (I - \mathbf{\tilde{H}}_j) \mathbf{x}_{j} \beta^{*}_{j} }{\mathbf{x}_{j}^T (I - \mathbf{\tilde{H}}_j) \mathbf{x}_{j} + (\tau^{\beta}_{kn})^{-2}}+ \frac{\mathbf{x}_{j}^T (I - \mathbf{\tilde{H}}_j) \mathbf{X}_{[-j]} \betavec^{*}_{[-j]}}{\mathbf{x}_{j}^T (I - \mathbf{\tilde{H}}_j) \mathbf{x}_{j} + (\tau^{\beta}_{kn})^{-2}}\\
     &=\beta^{*}_{j}-
     \frac{ \beta^{*}_{j} }{\tau^{2}_{kn}\mathbf{x}_{j}^T (I - \mathbf{\tilde{H}}_j) \mathbf{x}_{j} + 
     1}+ \frac{\mathbf{x}_{j}^T (I - \mathbf{\tilde{H}}_j) \mathbf{X}_{[-j]} \betavec^{*}_{[-j]}}{\mathbf{x}_{j}^T (I - \mathbf{\tilde{H}}_j) \mathbf{x}_{j} + (\tau^{\beta}_{kn})^{-2}}.\\
 \end{align}
 Now we arranged our theoretical results below. Detailed proofs are provided in appendix \ref{appendix}.
 We shall first discuss the positive semi-definiteness of the matrix $(\mathbf{I} - \widetilde{\mathbf{H}}_{-j})$. This is a very important step as it will ensure the Normal density of $p(\mathbf{y}|\beta_j,\mathbf{A}(\phivec), \sigma^2)$ and obtaining the optimizer $\hat{\beta}_j$ in equation \ref{beta_j_hat}.
\begin{Lemma}
\label{lem:psd_matrix}
Let $\mathbf{X}_{[-j]}(\phivec):=\mathbf{A}(\phivec)\mathbf{X}_{[-j]}$ be the design matrix with standardized
columns and dimension $n \times p_n$. Assume $\Sigma = \lim_{n\to\infty} \Sigma_{-j}(\phivec) =
\lim_{n\to\infty} \mathbf{X}_{[-j]}^T(\phivec) \mathbf{X}_{[-j]}(\phivec)/n$ is well defined and has all eigenvalues bounded away from 0 with minimum and maximum eigenvalues as $\lambda_{1}$ and $\lambda_{n}$. Additionally, assume the minimum and maximum eigenvalues of $\mathbf{A}^{T}(\phivec)\mathbf{A}(\phivec)$ are $d_{min}^{2}$ and $d_{max}^{2}$ such that $\inf_n \frac{n^{-1}\tau^{-2}_n}{\lambda_1^2 + n^{-1}\tau^{-2}_n} d_{\text{min}}^{2} $ and $\sup_n \frac{n^{-1}\tau^{-2}_n}{\lambda_n^2 + n^{-1}\tau^{-2}_n} d_{\text{max}}^{2}$ are in between 0 and 1. Then all the eigenvalues of $(\mathbf{I} - \widetilde{\mathbf{H}}_{-j})$ are bounded away from 0 when $\sup_n n\tau_n^2 <
\infty$.
\end{Lemma}

\noindent To prove the positive semi-definiteness of the matrix $(\mathbf{I} - \widetilde{\mathbf{H}}_{-j})$, we showed that this can be expressed as $X^T A X$ where $A$ is a PSD matrix. Now in Lemma \ref{lem:bound_by_eigen_value}, we showed that all the eigenvalues of the matrix $X^T A X$ can be upper and lower bounded by the minimum and maximum eigenvalues of $A$ and $X^T X$.

\begin{Lemma}
\label{lem:bound_by_eigen_value}
    Suppose \( A \) is an \( n \times n \) positive definite matrix, and \( X \) is an \( n \times p \) matrix. Let \( \lambda(A) \) denote an eigenvalue of \( A \). Let \( \lambda^*(X^T A X) \) and \( \lambda^*(X^T X) \) denote a non-zero eigenvalue of \( X^T A X \) and \( X^T X \), respectively. Then 
\[
\lambda_{\min}(A) \lambda^*_{\min}(X^T X) \leq \lambda^*(X^T A X) \leq \lambda_{\max}(A) \lambda^*_{\max}(X^T X),
\]
where 
\[
\lambda^*_{\min}(X^T X) = \max_{u \neq 0, Xu \neq 0} \frac{u^T X^T X u}{u^T u}, \quad \lambda^*_{\max}(X^T X) = \max_{u \neq 0, Xu \neq 0} \frac{u^T X^T X u}{u^T u}.
\]
Notice that \( A \) is positive definite, so \( \lambda_{\min}(A) > 0 \), and by definition, \( \lambda^*_{\min}(X^T X) > 0 \).
\end{Lemma}
To prove the Lemma \ref{lem:bound_by_eigen_value}, we need the help of a simple inequality in the Lemma \ref{lem:basic_inequality}.
\begin{Lemma}
\label{lem:basic_inequality}
    Let \( f(u) = g(u) h(u) \), \( u \in U \), where \( g(\cdot) \) and \( h(\cdot) \) are non-negative valued functions. Then 
\[
\left( \inf_{u \in U} g(u) \right) \left( \inf_{u \in U} h(u) \right) \leq f(u) \leq \left( \sup_{u \in U} g(u) \right) \left( \sup_{u \in U} h(u) \right) \quad \forall u \in U.
\]
Thus, it follows
\[
\inf_{u \in U} g(u) h(u) \geq \left( \inf_{u \in U} g(u) \right) \left( \inf_{u \in U} h(u) \right)
\]
and
\[
\sup_{u \in U} g(u) h(u) \leq \left( \sup_{u \in U} g(u) \right) \left( \sup_{u \in U} h(u) \right).
\]
\end{Lemma}
Now we want to obtain the upper bound in terms of $O(.)$ for the quadratic forms $\mathbf{x}^T (I - \mathbf{\tilde{H}}_j )\mathbf{x}$ and $\mathbf{x}^T_j (I - \tilde{H}_j) \mathbf{x}_k$ in Lemma \ref{lem:bounded_eigenvalue} and Lemma \ref{lem:order_for_different_index} respectively. 
\begin{Lemma}
\label{lem:bounded_eigenvalue}
    If $\sup_n n\tau^2_n < \infty$, then $n^{-1} \mathbf{x}^T (I - \mathbf{\tilde{H}}_j )\mathbf{x} = O(1)$ for any $n$-dimensional vector $\mathbf{x}$ with $\mathbf{x}^T \mathbf{x} = n$.
\end{Lemma}

\begin{Lemma}
\label{lem:order_for_different_index}
    Let $\Sigma$ has the elements $\rho_{kl}:=\frac{\mathbf{x}_{j}^{T}\mathbf{x}_{k}}{n}$. For any $j \in S^{*c}$ and $k \in S_{\beta}$, if
\[
\sup_{j \in S^{*c}} \max_{k \in S_{\beta}} 
\left\lvert 
\frac{\rho_{jk}}{1 + \sum_{\substack{l \neq j, \\ l \neq k}} \rho_{kl}}
\right\rvert = O\left(\sqrt{n}\tau^2_n\right)
\]
then $\mathbf{x}^T_j (I - \tilde{H}_j) \mathbf{x}_k = O(\sqrt{n})$.
\end{Lemma}
Finally, we want to show the convergence in probability of $\mu_{0j}$ to 0 and $\mu_{1j}$ to $\beta^*_j$ which ensures the spike of t distribution and slab of t distribution in the posterior distribution of $\beta_{j}$ given $\yvec$ is centralized properly.
\begin{Lemma}
\label{lem:l2_convergence}
    Given the same assumption of Lemmas \ref{lem:psd_matrix}, \ref{lem:bounded_eigenvalue}, \ref{lem:order_for_different_index}, and also $n(\tau^{\beta}_{0n})^2 \to 0$, $n(\tau^{\beta}_{1n})^2 \to \infty$ as $n \to \infty$, then $\mu_{0j} \xrightarrow{q.m.} 0$, and $\mu_{1j} \xrightarrow{q.m.} \beta^*_j$ where $\beta^*_j$ is the true parameter value for the $j$th covariate in (2), which implies $\mu_{0j} \xrightarrow{P} 0$, and $\mu_{1j} \xrightarrow{P} \beta^*_j$. Here $X_n \xrightarrow{q.m.} X$ (or $X_n \xrightarrow{L^2} X$) is defined by $\lim_{n \to \infty} \mathbb{E}[(X_n - X)^2] = 0$.
\end{Lemma}
In the model \ref{proposed_model_formula}, at each step ``t'', we estimate the correct active coefficients of $\betavec_{t}$ by lemma \ref{lem:l2_convergence}. Hence, $\epsilon_{t}=y_{t}-\nu -\sum_{j=1}^p \mathbf{x}_{t-j}^\prime\hat{\betavec_{j}}$ is correctly estimated and so is $\epsilonvec_{t}$. Now $\tilde{\epsilonvec}_{l}$ can be determined from $\epsilonvec_{t}$ by its construction. Now using the spike and slab prior we obtain the posterior mean estimates $\left(\hat{\phi}_{1},\hat{\phi}_{2},\cdots,\hat{\phi}_{q}\right)$ which also converges to their original value. The last statement is trivial as we use similar prior over the coefficients over $\phi_{i}$.

Now we want to concentrate on the posterior selection consistency of each of the variables. The true zero coefficients are expected to go to 0 whereas the non-zero coefficients should be away from zero when we calculate the posterior value of the activation indicator $Z_{j}$ for each of the coefficients $\beta_{j}$. 
We can follow the similar conditions mentioned in \cite{chen2019fast} on page 292. We first assume the relationship between the dimensions of the data. Condition~\ref{cond:1} helps to extend selection consistency to a high dimension. In our model, the number of predictors can be very high. Moreover, we also incorporate the lags of the covariates and dependent variable, which leads the time series model a very high dimensional problem. When n is the number of time points, $p_n$ is the number of variables in the existing model. We want the number of data points to grow such that the ratio of $\frac{\log(p_{n})}{n}$ should go to 0.  
\begin{condition}
\label{cond:1}
\( p_n = e^{n\delta_n }\) for some \( \delta_n \to 0 \) as \( n \to \infty \); that is, \( \frac{\log(p_n)}{n} \to 0 \). 
\end{condition}
Next, Condition~\ref{cond:2} states how the correlation between the two variables should behave when one predictor is from the active set and another predictor is from inactive set. We can write the index of the active variables as
\[
S_{\beta} \subset \{1, 2, \ldots, p_n\}, \quad \text{and} \quad S_{\beta} = \{ j : \beta^*_j \neq 0 \}.
\]
So the inactive set is $S_{\beta}^{c}= \{ j : \beta^*_j = 0 \}$. We saw in Lemma \ref{lem:order_for_different_index} that under Condition~\ref{cond:2}, we can control over the upper bound of the quadratic form for two different covariates, $\mathbf{x}_j$ from active set and $\mathbf{x}_k$ from the inactive set, $\mathbf{x}^T_j (I - \tilde{H}_j) \mathbf{x}_k$ as $O(\sqrt{n}\tau_{n}^{2})$. That is why, we need this condition for posterior pairwise selection consistency.
\begin{condition}
\label{cond:2}
\[
\sup_{j \in S_{\beta}^{c}} \max_{k \in S_{\beta}} \left\lvert \frac{\rho_{jk}}{1 + \sum_{l \neq j, l \neq k} \rho_{kl}} \right\rvert = O \left( \sqrt{n} \, \tau_{n}^2 \right)
\]
and
\[
\tau_n^2 = O \left( \frac{1}{n} \right),
\]
with \( n (\tau^{\beta}_{0n}) \to 0 \), \( n (\tau^{\beta}_{1n})^2 \to \infty \), \( \frac{\log(\tau^{\beta}_{1n} / \tau^{\beta}_{0n})}{n} \to 0 \) and \( q_{n}^{\beta} \sim p_n^{-1} \). Recall that $q^{\beta}_{n}$ is the prior probability of including a particular covariate in $\betavec$ and $p_n$ is the number of covariates with their lags based on the n time points.
\end{condition}
This is also important to mention that we can relax Condition~\ref{cond:2} when the number of covariates are less than the data points. The upper bound of $\tau_{n}^{2}$ is only required to properly define the matrix $\tilde{H}_j$ since $\mathbf{X}_{[-j]}^T(\phivec) \mathbf{X}_{[-j]}(\phivec)$ might not be invertible. For the case $p_{n}<n$, it is sufficient to consider the relaxed condition $\tau_n^2 = O \left( \frac{1}{\sqrt{n}} \right)$ as in that case $\sup_{j \in S_{\beta}^{c}} \max_{k \in S_{\beta}} \left\lvert \frac{\rho_{jk}}{1 + \sum_{l \neq j, l \neq k} \rho_{kl}} \right\rvert = O(1)$ and $\frac{w_{k}}{n}<\infty$ for $k\in S_{\beta}$ in Lemma \ref{lem:order_for_different_index}.

Finally, Condition~\ref{cond:3} states about the regularity of the scaled variance covariance matrix $\frac{\mathbf{X}(\phi)^T \mathbf{X}(\phi)}{n}$.
\begin{condition}
\label{cond:3}The maximum non-zero eigenvalues of the Gram matrix \( \frac{\mathbf{X}(\phi)^T \mathbf{X}(\phi)}{n} \) are bounded away from infinity.
\end{condition}
The posterior selection consistency theorems are given in Theorem~\ref{theorem:1} and Theorem~\ref{theorem:2}. Theorem~\ref{theorem:1} guarantees that the posterior selection of the true zero coefficients converges to 0. In case of any misspecification, either an inactive coefficient to be considered as non-zero or an active coefficient to be considered as 0, these probabilities go to zero when we have enough time points. This convergence rate depends on the spike and slab variance including the inclusion probability of each covariate.
\begin{theorem}
\label{theorem:1}
Let $\beta^*_j$ is the true value of the j th coefficient of $\betavec$, $\beta_j$.
Assume conditions \ref{cond:0}, \ref{cond:1}, \ref{cond:2} and \ref{cond:3} hold. Then there exists an increasing sequence \( d_n \) with \( \lim_{n \to \infty} d_n = d > 0 \), depending on the data, such that:

\begin{itemize}
    \item If \( \beta^*_j = 0 \), 
    \[
    \Pr(Z^{\beta}_{j} = 1 \mid Y) = O_P \left( q^{\beta}_{n}  \frac{\tau^{\beta}_{0n}}{\tau^{\beta}_{1n}} \right),
    \]
    \item If \( \beta^*_j \neq 0 \),
    \[
    \Pr(Z^{\beta}_{j} = 0 \mid Y) = O_P \left( \frac{1}{q^{\beta}_{n}} \frac{\tau^{\beta}_{1n}}{\tau^{\beta}_{0n}} e^{-n d_n} \right).
    \]
\end{itemize}

Therefore, 
\[
\label{2_probabilities}
\Pr \left( \Pr(Z^{\beta}_{j} = I\{\beta^*_j = 0\} \mid Y) > \epsilon \right) \to 0 \text{ for any } \epsilon > 0.
\]
\end{theorem}

\begin{proof}
Conditioning of the value of $\phivec$, the proof of this theorem is inspired from \cite{chen2019fast} on page 293, but it heavily depends the Lemma \ref{lem:psd_matrix} - Lemma \ref{lem:l2_convergence}, which demonstrates the related results and required conditions for the AR process. 
\end{proof}
Theorem \ref{theorem:1} guarantees the misspecification for individual covariates goes to 0 when we have a big set of time points. Now we want to talk about the selection consistency for the true model, i.e. how to identify the true model with each of the covariates with their corresponding inclusion or exclusion in the final model. Let $\text{T}^{\betavec}$ be a binary vector where active variables for $\betavec$ correspond to 1 and for inactive variables in $\betavec$, it is 0. Theorem \ref{theorem:2} describes that under Conditions \ref{cond:0}, \ref{cond:1}, \ref{cond:2}, and \ref{cond:3}, the selection of the true model goes to 1 when n is large. There are two probabilities mentioned in the expression \ref{2_probabilities}. The inner probability is the posterior probability, a statistic conditioned on the data, whereas the outer probability considers the distribution of the data, which is generated from the true model \ref{proposed_model}.
\begin{theorem}
\label{theorem:2}
 Assume Conditions \ref{cond:0}, \ref{cond:1}, \ref{cond:2}, and \ref{cond:3},  hold. For the model \ref{final_model} with the specifications in equations \ref{prior_specification_beta}, \ref{prior_specification_phi}, \ref{selection_probability}, and \ref{prior_to_other_coeffficients}, we have 
\[
\Pr(Z^{\betavec} = T^{\betavec} \mid Y) \xrightarrow{P} 1 \text{ as } n \to \infty,
\]
that is, the posterior probability of the true model goes to 1 as the time points increase to \( \infty \). In particular, for any \( 0 < \epsilon < 1 \),
\[
\Pr \left[ \Pr(Z^{\beta}_{j} = T_{j}^{\betavec} \mid Y) > \epsilon \text{ for all } j = 1, \ldots, p_n \right] \geq 1 - O \left( \frac{\tau^{\beta}_{0n}}{\tau^{\beta}_{1n}} \right) \to 1 \text{ as } n \to \infty.
\] 
\end{theorem}

\begin{proof}
Conditioning of the value of $\phivec$, the proof is inspired from \cite{chen2019fast} outlined on page 295. This proof also depends on the Lemma \ref{lem:psd_matrix} - Lemma \ref{lem:l2_convergence}, which demonstrates the related results and required conditions for the AR process in the presence of lagged covariates and structure of the error process.
\end{proof}

From the assumptions of Lemma~\ref{lem:l2_convergence}, the estimates of \(\betavec_{j}\) converge to their true values, ensuring the correct selection of active and inactive coefficients \(Z^{\betavec}\) as stated in Theorem~\ref{theorem:2}. Define \(\hat{\epsilon}_{t} = y_t - \nu - \sum_{j=1}^p \mathbf{x}_{t-j}^\prime \hat{\mu}_{j}\). For sufficient large ``t'' in the true assumed Model~\ref{proposed_model_formula}, 
\[
\hat{\epsilon}_{t} = \epsilon_{t} + o_{p}(1).
\]

So we can work with $\epsilon_{t}$. Let us define,
\[
\boldsymbol{\epsilon}_{n} = \begin{pmatrix} \epsilon_1 \\ \epsilon_2 \\ \vdots \\ \epsilon_n \end{pmatrix}, \quad
\tilde{\boldsymbol{\epsilon}}_l = \begin{pmatrix} \mathbf{0}_l \\ \boldsymbol{\epsilon}_{n-l} \end{pmatrix}, \quad
\boldsymbol{\epsilon}_{n}^{*}= \begin{pmatrix}  \epsilon_2 \\ \vdots \\ \epsilon_n \end{pmatrix}
.
\]

Let \(\mathbf{E} \in \mathbb{R}^{n\times q}\) be a matrix with columns \(\tilde{\boldsymbol{\epsilon}}_1, \tilde{\boldsymbol{\epsilon}}_2, \dots, \tilde{\boldsymbol{\epsilon}}_q\). Hence,
\[
\mathbf{E}=
\begin{pmatrix}
    \mathbf{0}^{n}\\ \mathbf{E}_{n-1}
\end{pmatrix}
 =
\begin{pmatrix}
\tilde{\boldsymbol{\epsilon}}_1, \tilde{\boldsymbol{\epsilon}}_2, \dots, \tilde{\boldsymbol{\epsilon}}_q
\end{pmatrix}
=
\begin{pmatrix}
0 & 0 & \cdots & 0 \\
\epsilon_{1} & 0 & \cdots & 0 \\
\epsilon_{2} & \epsilon_{1} & \cdots & 0 \\
\vdots & \vdots & \ddots & \vdots \\
\epsilon_{n-1} & \epsilon_{n-2} & \cdots & \epsilon_{n-q}
\end{pmatrix}.
\]
From Model~\ref{proposed_model_formula}, we can write
\[
\boldsymbol{\epsilon}_{n}=\mathbf{E}\phivec+\mathbf{u}_{n};\quad
\mathbf{u}_{n} \sim N(\mathbf{0},\sigma^2\mathbf{I}).
\]
The above can also be written as,
\[
 \mathbf{A}(\phivec) \boldsymbol{\epsilon}_{n}=\mathbf{u}_{n}, \quad
 \boldsymbol{\epsilon}_{n} \sim N(\mathbf{0},\mathbf{\Sigma}(\phivec,\sigma^2)=\sigma^2 \mathbf{A}(\phivec)^{-1} (\mathbf{A}^{-1}(\phivec))^{T}).
\]
Now as the first row of the matrix $\mathbf{E}$ a $\mathbf{0}$ vector, we can remove that and simplify it as,
\[
\boldsymbol{\epsilon}_{n}^{*}=\mathbf{E}_{n-1}\phivec+\mathbf{u}_{n-1}.
\]

Two conditions \ref{cond:2_phi} and \ref{cond:3_phi}, similar to conditions \ref{cond:2} and \ref{cond:3}, are required to establish the posterior consistency of \(\phivec\). Let the matrix $\mathbf{E}_{n-1}$ has $n=n-1$ many rows and 
the sample variance-covariance matrix constructed by the errors with standardized columns defined as,
\[
\tilde{\mathbf{\Sigma}}=\mathbb{E}\left[\frac{\mathbf{E}_{n-1}^{T}\mathbf{E}_{n-1}}{n-1}\right].
\]
Let $\tilde{\mathbf{\Sigma}}$ has the elements $\rho_{jk}^{\phivec}:=\frac{\mathbf{e}_{j}^{T}\mathbf{e}_{k}}{n-1}$ where the j-th and k-th columns of $\tilde{\mathbf{\Sigma}}$ are defined as $\mathbf{e}_{j}$ and $\mathbf{e}_{k}$. As $O(\frac{n-1}{n})=1$, we write the conditions for the selection consistency of $\phi$ with the notation n instead of n-1 to simplify the notation.
\begin{condition}
\label{cond:2_phi}
\[
\sup_{j \in S_{\phivec}^{c}} \max_{k \in S_{\phivec}} \left\lvert \frac{\rho_{jk}^{\phivec}}{1 + \sum_{l \neq j, l \neq k} \rho_{kl}^{\phivec}} \right\rvert = O \left( \sqrt{n} \, \tau_{n}^2 \right)
\]
and
\[
\tau_{n}^2 = O \left( \frac{1}{\sqrt{n}} \right),
\]
with \( n (\tau^{\phi}_{0n}) \to 0 \), \( n (\tau^{\phi}_{1n})^2 \to \infty \), \( \frac{\log(\tau^{\phi}_{1n} / \tau^{\phi}_{0n})}{n} \to 0 \) and \( q^{\phi}_{n} \sim q_{n}^{-1} \). Recall that $q^{\phi}_{n}$ is the prior probability of including a lag in $\phivec$ and $q_n$ is the number of lags based on the n sample.
\end{condition}
Similar to Lemma \ref{lem:order_for_different_index} one can show that under Condition~\ref{cond:2_phi}, we can control over the upper bound of the quadratic form for two different lag of error vector, $\mathbf{e}_j$ from the active set and $\mathbf{e}_k$ from the inactive set, $\mathbf{e}^T_j (I - \tilde{H}_{j}^{\phivec}) \mathbf{e}_k$ as $O(\sqrt{n}\tau_{n}^{2})$, where $\tilde{H}_{j}^{\phivec} = \mathbb{E}\left[\mathbf{E}_{[-j]} 
\left( \mathbf{E}_{[-j]}^T \mathbf{E}_{[-j]} + \tau^{-2}_n \mathbf{I} \right)^{-1} 
\mathbf{E}_{[-j]}^T\right]$.
\begin{condition}
\label{cond:3_phi}The maximum non-zero eigenvalues of the Gram matrix \( \tilde{\mathbf{\Sigma}}\) are bounded away from infinity. 
\end{condition}
Recall from Lemma~\ref{lem:l2_convergence}, $\mathcal{A}^{\betavec}_{n,\epsilon}:= \left\{\hat{\mu}_{j} \xrightarrow{\text{a.s.}} \betavec_{j} : Z^{\betavec} = T^{\betavec}; 
\forall j = 1, \dots, p_{n} \mid Y \right\}$, becomes probability 1 set as $n$ grows to $\infty$.
\begin{theorem}
\label{theorem:1_phi}
We defined $\phi^*_j$ is the true value of the j th coefficient of $\phivec$, $\phi_j$.
Assume conditions of Theorem~\ref{theorem:2}, Condition~\ref{cond:2_phi} and \ref{cond:3_phi}. Then there exists an increasing sequence \( d_{n} \) with \( \lim_{n \to \infty} d_{n} = d > 0 \), depending on the data, such that:

\begin{itemize}
    \item If \( \phi^*_j = 0 \), 
    \[
    \Pr(Z^{\phivec}_{j} = 1 \mid Y,\mathcal{A}^{\betavec}_{n,\epsilon}) = O_P \left( q_{n}^{\phi}  \frac{\tau^{\phi}_{0n}}{\tau^{\phi}_{1n}} \right),
    \]
    \item If \( \phi^*_j \neq 0 \),
    \[
    \Pr(Z^{\phivec}_{j} = 0 \mid Y,\mathcal{A}^{\betavec}_{n,\epsilon}) = O_P \left( \frac{1}{q_{n}^{\phi}} \frac{\tau^{\phi}_{1n}}{\tau^{\phi}_{0n}} e^{-n d_{n}} \right).
    \]
\end{itemize}

Therefore, 
\[
\Pr \left( \Pr(Z^{\phivec}_{j} = I\{\phivec^*_j = 0\} \mid Y,\mathcal{A}^{\betavec}_{n,\epsilon}) > \epsilon \right) \to 0 \text{ for any } \epsilon > 0.
\]
\end{theorem}

\begin{proof}
    This proof depends on the Lemma \ref{lem:psd_matrix} - Lemma \ref{lem:l2_convergence}, which demonstrates the related results and required conditions for the AR process in the presence of lagged covariates and structure of the error process. After we estimated the coefficients of the covariates, conditioned on that, this proof can be done inspired by \cite{chen2019fast} demonstrated on page 295.
\end{proof}

\begin{theorem}
\label{theorem:2_phi}
 Assume Conditions of Theorem~\ref{theorem:2}, Condition~\ref{cond:2_phi} and \ref{cond:3_phi}. For the model \ref{proposed_model_formula} with the specifications in equations \ref{prior_specification_beta}, \ref{prior_specification_phi}, \ref{selection_probability}, and \ref{prior_to_other_coeffficients}, we have 
\[
\Pr(Z^{\phivec} = T^{\phivec} \mid Y, \mathcal{A}^{\betavec}_{n,\epsilon}) \xrightarrow{P} 1 \text{ as } n \to \infty,
\]
that is, the posterior probability of the true model goes to 1 as the time points increase to \( \infty \). In particular, for any \( 0 < \epsilon < 1 \),
\[
\Pr \left[ \Pr(Z^{\phivec}_{j} = T_{j}^{\phivec} \mid Y,\mathcal{A}^{\betavec}_{n,\epsilon}) > \epsilon \text{ for all } j = 1, \ldots, q \right] \geq 1 - O \left( \frac{\tau^{\phi}_{0n}}{\tau^{\phi}_{1n}} \right) \to 1 \text{ as } n \to \infty.
\] 
\end{theorem}

\begin{proof}
This proof is inspired by \cite{chen2019fast} on page 295. We needed additional Lemma \ref{lem:psd_matrix} - Lemma \ref{lem:l2_convergence}, which demonstrates the related results and required conditions for the AR process in the presence of lagged covariates and the structure of the error process.
\end{proof}

\section{Estimation procedure}
\label{estimation_procedure}
The BoomSpikeSlab package in R (\cite{scott2023package}) offers a flexible and efficient framework for Bayesian variable selection in regression settings, particularly those involving a large number of potentially relevant predictors. It implements spike-and-slab priors, a widely used class of shrinkage priors that facilitates sparse model selection by distinguishing between signal and noise through a mixture of a point mass at zero (the spike) and a continuous distribution (the slab). The package supports various regression models, including linear, logistic, Poisson, multinomial logit, quantile, and Student’s t-regression, thereby accommodating both Gaussian and non-Gaussian error structures.

One of the main advantages of BoomSpikeSlab is its seamless integration with Markov Chain Monte Carlo (MCMC) methods for posterior inference, allowing users to obtain posterior inclusion probabilities, estimate credible intervals for model parameters, and evaluate model uncertainty in a principled Bayesian framework. This makes it particularly useful in high-dimensional contexts where traditional frequentist methods may fail to provide stable or interpretable solutions.

In our study, we adopt a two-stage procedure leveraging the capabilities of BoomSpikeSlab to perform simultaneous variable selection in the context of autoregressive models with lagged covariates and autocorrelated errors:
\begin{itemize}
\item[Stage 1] – {\bf Selection of Covariates and Their Lags}: We first apply the BoomSpikeSlab framework to a linear regression model that includes both the current and lagged values of exogenous predictors. The spike-and-slab prior facilitates the identification of a sparse subset of influential covariates and their delayed effects on the response variable.

\item[Stage 2] – {\bf Selection of Lagged Residuals in an AR Framework}: Using the residuals obtained from the first-stage model, we fit a secondary autoregressive model where the lagged residuals are treated as candidate predictors. This step allows us to capture the temporal dependence structure in the error process. Once again, we employ spike-and-slab priors to select the significant lag terms in the AR process.
\end{itemize}

This two-step methodology provides a principled and computationally feasible approach to modeling both exogenous and endogenous dependencies in time series data. The use of BoomSpikeSlab in both stages ensures that uncertainty is properly quantified and that variable selection is data-driven, robust, and interpretable. This framework is particularly valuable in domains such as finance, environmental sciences, and engineering, where lagged effects and autocorrelation are prominent.


Now we illustrate the prediction procedure for horizon $h$ and the $n$-th step. 

a) Select a window $n=T_{0}$ and regress $\mathbf{y}_{n} = (y_{1},\cdots,y_{n})$ over the corresponding variables from $\mathbf{X}$; which consists the covariates and their lag with a spike and slab prior over the coefficients. For our experiments, we used $\frac{2}{3}$ of the whole data as the initial training window. While using the BoomSpikeSlab package, we obtain the inclusion probability for each coefficient in the model. As we have initially ``p'' many coefficients, the probability of selecting one coefficient uniformly is $\frac{1}{p}$. We select the coefficients with the inclusion probability greater than $\frac{1}{p}$ and create an initial set $\text{T}^{\betavec}$; a set of active coefficients. Based on the estimate of the coefficients $\hat{\betavec}$, We obtain the estimates of the error process;
\[
\hat{\epsilonvec}_{n}=\mathbf{y}_{n} - \hat{\mathbf{y}}_{n}
\]
 and predict $\left(\hat{y}_{n+1},\hat{y}_{n+2},\cdots,\hat{y}_{n+h}\right)$, i.e. $\hat{y}_{n+j}:=\tilde{\mathbf{x}}_{n+j}^T\hat{\betavec}$ for $j=1,2,\cdots,h$, the posterior mean based on the selected features. This is the stage 1 of the procedure.

b) The actual number of lags will be unknown in the error process. So we choose a value of ``q'' which is sufficiently large; for example, we considered it here as 10. Now regress $\epsilonvec_{n}$ over the matrix $\mathbf{E}$ with columns \(\tilde{\boldsymbol{\epsilon}}_1, \tilde{\boldsymbol{\epsilon}}_2, \dots, \tilde{\boldsymbol{\epsilon}}_q\) and obtain the initial estimate $\hat{\phivec}$. As we have initially ``q'' many lags, the probability of selecting one lag uniformly is $\frac{1}{q}$. We select the coefficients with the inclusion probability greater than $\frac{1}{q}$ and create an initial set $\text{T}^{\phivec}$; a set of active lags. This is the stage 2 of the procedure.

c) Using the estimated $\hat{\phivec}$ and the error sequence estimate previously, we calculate $\hat{\epsilon}_{n+1},\cdots,\hat{\epsilon}_{n+h}$ as a rolling basis \[\hat{\epsilon}_{n+k}\leftarrow\sum_{i=1}^{q}\hat{\phi}_{i}\hat{\epsilon}_{n+k-i}\] for $k=1,2,\cdots,h$
and use them to update \[\hat{y}_{n+j}\leftarrow\hat{y}_{n+j}+\hat{\epsilon}_{n+j}; j=1,\cdots,h.\]

Now we vary the window n to n+1 and continue for the above procedure. Note that the set $\text{T}^{\betavec}$, $\text{T}^{\phivec}$, $\hat{\betavec}$ and $\hat{\phivec}$ are also going to change in the rolling process.

\section{Numerical Illustrations using Simulated Data}
\label{simulation_results}
We generate $N=500$ multivariate normal observations and construct the design matrix $\mathbf{X}_{N\times p}$. We take $\betavec=\left(3, -3, 1, -1, 0.5,\mathbf{0}_{p-5}\right)$ and $\phivec=\left(0.9, -0.9, 0.5, -0.5,\mathbf{0}_{q-4}\right)$. We chose $p=50$ and $q=10$. Next, we generate observations 
$\mathbf{y}=\mathbf{X}\betavec+\epsilonvec$,
where the error $\epsilonvec$ is generated from an autoregressive process with a true sigma, $\sigma=1,2,3$. We repeated the whole experiment 50 times using the complete data N to observe the selection accuracy for the variables and errors' lag. For the simulation, we did not perform the rolling basis updates to avoid the computational complexity. Let us define some common metrics for comparison of the model performance. TP (True Positive) is the number of instances where the model correctly predicts the true nonzero coefficients in $\betavec$ and $\phivec$ whereas TN (True Negative) is the number of cases where the model correctly predicts the true zero coefficients in $\betavec$ and $\phivec$. The number of instances where the model incorrectly predicts the non-zero coefficients is determined by FP (False Positives). The number of cases where the model incorrectly predicts the zero coefficients is determined by FN (False Negative). 

Let us assume the initial window of the model is $n=T_{0}$ and then we perform the prediction for the rest of time series values n+1 to N based on the algorithm in a rolling basis for the real data analysis. The Mean Squared Error (MSE) quantifies the average squared difference between the observed values (\(y_i\)) and the predicted values (\(\hat{y}_i\)). 

MSE is defined as:
\[
\text{MSE} = \frac{1}{N-n} \sum_{i=n+1}^N (y_i - \hat{y}_i)^2,
\]
where \(y_i\) is observed values of the response variable, \(\hat{y}_i\) is predicted values of the response variable, and \(n\) is number of observations in the dataset. Lower MSE values indicate better model performance, as they reflect smaller prediction errors. The Pearson correlation coefficient (\(r\)) measures the linear relationship between the observed (\(y_i\)) and predicted values (\(\hat{y}_i\)). It is calculated as:
\[
r = \frac{\sum_{i=n+1}^N (y_i - \bar{y})(\hat{y}_i - \bar{\hat{y}})}{\sqrt{\sum_{i=n+1}^N (y_i - \bar{y})^2 \sum_{i=n+1}^N (\hat{y}_i - \bar{\hat{y}})^2}},
\] where \(\bar{y}\) is the mean of the observed values. The Coefficient of Determination (\(R^2\)) assesses the proportion of variance in the observed values that is explained by the predictive model. It is defined as:
\[
R^2 = 1 - \frac{\sum_{i=n+1}^N (y_i - \hat{y}_i)^2}{\sum_{i=n+1}^N (y_i - \bar{y})^2},
\]
where \(\sum_{i=n+1}^N (y_i - \hat{y}_i)^2\) is the Residual Sum of Squares (RSS), representing unexplained variance and \(\sum_{i=n+1}^N (y_i - \bar{y})^2\) is Total Sum of Squares (TSS), representing total variance in the observed data.
An \(R^2\) value close to 1 indicates that the model explains a large proportion of the variance in the response variable, while a value close to 0 indicates poor explanatory power. Mean Error (ME) measures the average bias in the predictions. It is calculated as the mean of the differences between the observed values (\(y_i\)) and the predicted values (\(\hat{y}_i\)), given by:
\[
\text{ME} = \frac{1}{N-n} \sum_{i=n+1}^N (y_i - \hat{y}_i).
\] The ME indicates whether the predictions are systematically overestimating or underestimating the observed values. A positive ME suggests a tendency to underestimate, while a negative ME indicates overestimation. A value close to 0 suggests unbiased predictions. Mean Absolute Error (MAE) quantifies the average magnitude of prediction errors without considering their direction. It is defined as:
\[
\text{MAE} = \frac{1}{N-n} \sum_{i=n+1}^N |y_i - \hat{y}_i|.
\] Unlike ME, the MAE always yields a non-negative value, as it focuses solely on the magnitude of the errors. Lower MAE values indicate better model accuracy, as they reflect smaller deviations between predictions and actual observations. While both metrics assess prediction accuracy, ME is sensitive to the direction of errors and provides insights into bias, whereas MAE provides a robust measure of average error magnitude. Together, these metrics offer complementary perspectives on model performance.

\subsection{Variable selection accuracy}\label{sim_1}
For each simulation, we repeat the experiment 50 times and compare the performance through the mean accuracy $\frac{\text{TP + TN}}{\text{TP + FP + FN + TN}}$ for different values of the true sigma. 
The threshold is selected for $\betavec$ as $.001\min_{j} \{|\beta_j|: \beta_j \neq 0\}$ and $\phivec$ as $.001 \min_{i} \{|\phi_i|: \phi_i\neq 0\}$. The simulation results are in Table~\ref{tab:simulation_beta} and \ref{tab:simulation_phi}.

\begin{table}[ht]
\centering
\caption{TP, FP, FN, TN, and Accuracy for $\betavec$} 
\scalebox{0.75}{
\begin{tabular}{|c|c|p{3.5cm}|p{3.5cm}|p{3.5cm}|p{3.5cm}|p{3.5cm}|}
  \hline
$q$ & $\sigma$ & $p=30$ & $p=35$ & $p=40$ & $p=45$ & $p=50$ \\ 
  \hline
\multirow{3}{*}{10} & 1 & 4.56, 0.02, 0.44, 24.98, 0.98 & 4.52, 0.04, 0.48, 29.96, 0.99 & 4.62, 0.08, 0.38, 34.92, 0.99 & 4.58, 0, 0.42, 40, 0.99 & 4.64, 0.04, 0.36, 44.96, 0.99 \\ 
\cline{3-7}
   & 2 & 4.58, 0, 0.42, 25, 0.99 & 4.52, 0.02, 0.48, 29.98, 0.99 & 4.54, 0, 0.46, 35, 0.99 & 4.58, 0.10, 0.42, 39.9, 0.99 & 4.46, 0, 0.54, 45, 0.99 \\ 
\cline{3-7}
   & 3 & 4.44, 0.02, 0.56, 24.98, 0.98 & 4.52, 0.10, 0.48, 29.9, 0.98 & 4.56, 0.04, 0.44, 34.96, 0.99 & 4.36, 0.06, 0.64, 39.94, 0.98 & 4.50, 0.08, 0.50, 44.92, 0.99 \\ 
  \hline
\multirow{3}{*}{20} & 1 & 4.52, 0.02, 0.48, 24.98, 0.98 & 4.56, 0.06, 0.44, 29.94, 0.99 & 4.48, 0.04, 0.52, 34.96, 0.99 & 4.60, 0.02, 0.40, 39.98, 0.99 & 4.52, 0.06, 0.48, 44.94, 0.99 \\ 
\cline{3-7}
   & 2 & 4.62, 0, 0.38, 25, 0.99 & 4.48, 0.02, 0.52, 29.98, 0.98 & 4.60, 0.08, 0.40, 34.92, 0.99 & 4.62, 0.10, 0.38, 39.9, 0.99 & 4.40, 0.02, 0.60, 44.98, 0.99 \\ 
\cline{3-7}
 & 3 & 4.50, 0, 0.50, 25, 0.98 & 4.54, 0.02, 0.46, 29.98, 0.99 & 4.60, 0.06, 0.40, 34.94, 0.99 & 4.58, 0, 0.42, 40, 0.99 & 4.54, 0.02, 0.46, 44.98, 0.99 \\ 
  \hline
\multirow{3}{*}{30} & 1 & 4.54, 0.04, 0.46, 24.96, 0.98 & 4.62, 0.04, 0.38, 29.96, 0.99 & 4.52, 0.02, 0.48, 34.98, 0.99 & 4.48, 0.04, 0.52, 39.96, 0.99 & 4.54, 0.10, 0.46, 44.9, 0.99 \\
\cline{3-7}
   & 2 & 4.52, 0, 0.48, 25, 0.98 & 4.40, 0.06, 0.60, 29.94, 0.98 & 4.64, 0.04, 0.36, 34.96, 0.99 & 4.50, 0.04, 0.50, 39.96, 0.99 & 4.60, 0.04, 0.40, 44.96, 0.99 \\ 
\cline{3-7}
   & 3 & 4.64, 0.02, 0.36, 24.98, 0.99 & 4.66, 0.02, 0.34, 29.98, 0.99 & 4.54, 0.02, 0.46, 34.98, 0.99 & 4.52, 0.02, 0.48, 39.98, 0.99 & 4.54, 0.06, 0.46, 44.94, 0.99 \\ 
  \hline
\multirow{3}{*}{40} & 1 & 4.58, 0, 0.42, 25, 0.99 & 4.60, 0.06, 0.40, 29.94, 0.99 & 4.50, 0, 0.50, 35, 0.99 & 4.62, 0.02, 0.38, 39.98, 0.99 & 4.48, 0.02, 0.52, 44.98, 0.99 \\ 
\cline{3-7}
   & 2 & 4.62, 0.02, 0.38, 24.98, 0.99 & 4.64, 0.04, 0.36, 29.96, 0.99 & 4.54, 0, 0.46, 35, 0.99 & 4.58, 0.06, 0.42, 39.94, 0.99 & 4.48, 0.04, 0.52, 44.96, 0.99 \\ 
\cline{3-7}
   & 3 & 4.48, 0.10, 0.52, 24.9, 0.98 & 4.50, 0, 0.50, 30, 0.99 & 4.54, 0, 0.46, 35, 0.99 & 4.58, 0.08, 0.42, 39.92, 0.99 & 4.50, 0.10, 0.50, 44.9, 0.99 \\ 
  \hline
\multirow{3}{*}{50} & 1 & 4.50, 0.02, 0.50, 24.98, 0.98 & 4.50, 0.02, 0.50, 29.98, 0.99 & 4.62, 0.02, 0.38, 34.98, 0.99 & 4.52, 0.04, 0.48, 39.96, 0.99 & 4.54, 0.06, 0.46, 44.94, 0.99 \\
\cline{3-7}
   & 2 & 4.40, 0.02, 0.60, 24.98, 0.98 & 4.54, 0, 0.46, 30, 0.99 & 4.62, 0.04, 0.38, 34.96, 0.99 & 4.62, 0.08, 0.38, 39.92, 0.99 & 4.60, 0.04, 0.40, 44.96, 0.99 \\ 
\cline{3-7}
   & 3 & 4.46, 0.02, 0.54, 24.98, 0.98 & 4.58, 0.02, 0.42, 29.98, 0.99 & 4.50, 0.06, 0.50, 34.94, 0.99 & 4.62, 0.04, 0.38, 39.96, 0.99 & 4.60, 0.04, 0.40, 44.96, 0.99 \\ 
  \hline
\end{tabular}}
\label{tab:simulation_beta}
\end{table}

\begin{table}[ht]
\centering
\caption{TP, FP, FN, TN, and Accuracy for $\phivec$} 
\scalebox{.75}{
\begin{tabular}{|c|c|p{3.5cm}|p{3.5cm}|p{3.5cm}|p{3.5cm}|p{3.5cm}|}
  \hline
$q$ & $\sigma$ & $p=30$ & $p=35$ & $p=40$ & $p=45$ & $p=50$ \\ 
  \hline
\multirow{3}{*}{10.00} & 1 & 3.84, 0.4, 0.16, 5.6, 0.94 & 3.98, 0.3, 0.02, 5.7, 0.97 & 3.9, 0.34, 0.1, 5.66, 0.96 & 3.88, 0.32, 0.12, 5.68, 0.96 & 3.9, 0.28, 0.1, 5.72, 0.96 \\ \cline{3-7}
   & 2 & 3.86, 0.3, 0.14, 5.7, 0.96 & 3.84, 0.46, 0.16, 5.54, 0.94 & 3.94, 0.3, 0.06, 5.7, 0.96 & 3.9, 0.32, 0.1, 5.68, 0.96 & 3.96, 0.28, 0.04, 5.72, 0.97 \\ \cline{3-7}
   & 3 & 3.92, 0.42, 0.08, 5.58, 0.95 & 3.96, 0.38, 0.04, 5.62, 0.96 & 3.8, 0.32, 0.2, 5.68, 0.95 & 3.94, 0.32, 0.06, 5.68, 0.96 & 3.92, 0.24, 0.08, 5.76, 0.97 \\ 
  \hline
\multirow{3}{*}{20.00} & 1 & 3.94, 0.3, 0.06, 15.7, 0.98 & 3.92, 0.38, 0.08, 15.62, 0.98 & 3.8, 0.54, 0.2, 15.46, 0.96 & 3.86, 0.46, 0.14, 15.54, 0.97 & 3.78, 0.64, 0.22, 15.36, 0.96 \\ \cline{3-7}
   & 2 & 3.94, 0.36, 0.06, 15.64, 0.98 & 4, 0.32, 0, 15.68, 0.98 & 3.92, 0.26, 0.08, 15.74, 0.98 & 3.92, 0.28, 0.08, 15.72, 0.98 & 3.96, 0.3, 0.04, 15.7, 0.98 \\ \cline{3-7}
   & 3 & 3.92, 0.36, 0.08, 15.64, 0.98 & 3.94, 0.4, 0.06, 15.6, 0.98 & 3.92, 0.22, 0.08, 15.78, 0.99 & 3.94, 0.34, 0.06, 15.66, 0.98 & 3.88, 0.32, 0.12, 15.68, 0.98 \\ 
  \hline
\multirow{3}{*}{30.00} & 1 & 3.94, 0.4, 0.06, 25.6, 0.98 & 3.94, 0.34, 0.06, 25.66, 0.99 & 3.94, 0.34, 0.06, 25.66, 0.99 & 3.86, 0.5, 0.14, 25.5, 0.98 & 3.94, 0.6, 0.06, 25.4, 0.98 \\ \cline{3-7}
   & 2 & 3.88, 0.28, 0.12, 25.72, 0.99 & 3.86, 0.32, 0.14, 25.68, 0.98 & 3.96, 0.44, 0.04, 25.56, 0.98 & 3.92, 0.32, 0.08, 25.68, 0.99 & 3.92, 0.48, 0.08, 25.52, 0.98 \\ \cline{3-7}
   & 3 & 3.88, 0.32, 0.12, 25.68, 0.99 & 3.88, 0.42, 0.12, 25.58, 0.98 & 3.9, 0.38, 0.1, 25.62, 0.98 & 3.92, 0.46, 0.08, 25.54, 0.98 & 3.86, 0.3, 0.14, 25.7, 0.99 \\ 
  \hline
\multirow{3}{*}{40.00} & 1 & 3.88, 0.44, 0.12, 35.56, 0.99 & 3.94, 0.44, 0.06, 35.56, 0.99 & 3.92, 0.5, 0.08, 35.5, 0.99 & 3.92, 0.3, 0.08, 35.7, 0.99 & 3.88, 0.44, 0.12, 35.56, 0.99 \\ \cline{3-7}
   & 2 & 3.92, 0.34, 0.08, 35.66, 0.99 & 3.9, 0.44, 0.1, 35.56, 0.99 & 3.98, 0.38, 0.02, 35.62, 0.99 & 3.84, 0.42, 0.16, 35.58, 0.99 & 3.88, 0.6, 0.12, 35.4, 0.98 \\ \cline{3-7}
   & 3 & 3.96, 0.34, 0.04, 35.66, 0.99 & 3.88, 0.58, 0.12, 35.42, 0.98 & 3.9, 0.32, 0.1, 35.68, 0.99 & 3.96, 0.3, 0.04, 35.7, 0.99 & 3.88, 0.56, 0.12, 35.44, 0.98 \\ 
  \hline
\multirow{3}{*}{50.00} & 1 & 3.8, 0.42, 0.2, 45.58, 0.99 & 3.96, 0.3, 0.04, 45.7, 0.99 & 3.98, 0.3, 0.02, 45.7, 0.99 & 3.92, 0.56, 0.08, 45.44, 0.99 & 3.92, 0.36, 0.08, 45.64, 0.99 \\ \cline{3-7}
   & 2 & 3.82, 0.5, 0.18, 45.5, 0.99 & 3.94, 0.44, 0.06, 45.56, 0.99 & 3.9, 0.3, 0.1, 45.7, 0.99 & 3.92, 0.56, 0.08, 45.44, 0.99 & 3.88, 0.46, 0.12, 45.54, 0.99 \\ \cline{3-7}
   & 3 & 3.88, 0.28, 0.12, 45.72, 0.99 & 3.94, 0.34, 0.06, 45.66, 0.99 & 3.88, 0.38, 0.12, 45.62, 0.99 & 3.9, 0.46, 0.1, 45.54, 0.99 & 3.98, 0.36, 0.02, 45.64, 0.99 \\ 
   \hline
\end{tabular}}
\label{tab:simulation_phi}
\end{table}

\subsection{Prediction accuracy}
As we can observe the average selection accuracy is pretty good, 93\% to 98\% for $\beta$ and $\phi$, we may be interested to observe the prediction accuracy. We choose $N=1000$, $p=50$, and $q=10$ with the true values and an autoregressive model as mentioned in section \ref{sim_1}. We train our model based on the first $n=800$ observations and then test based on the last $N-n=200$ observations. The results are presented in the Table~\ref{tab:sim_prediction_accuracy}. We obtain the best model fit $R^{2}_{\text{adj}}=.75$ and the lowest MSE 2.23 for $h=5$.
\begin{figure}[h!]
    \centering
    \includegraphics[width=0.7\textwidth]{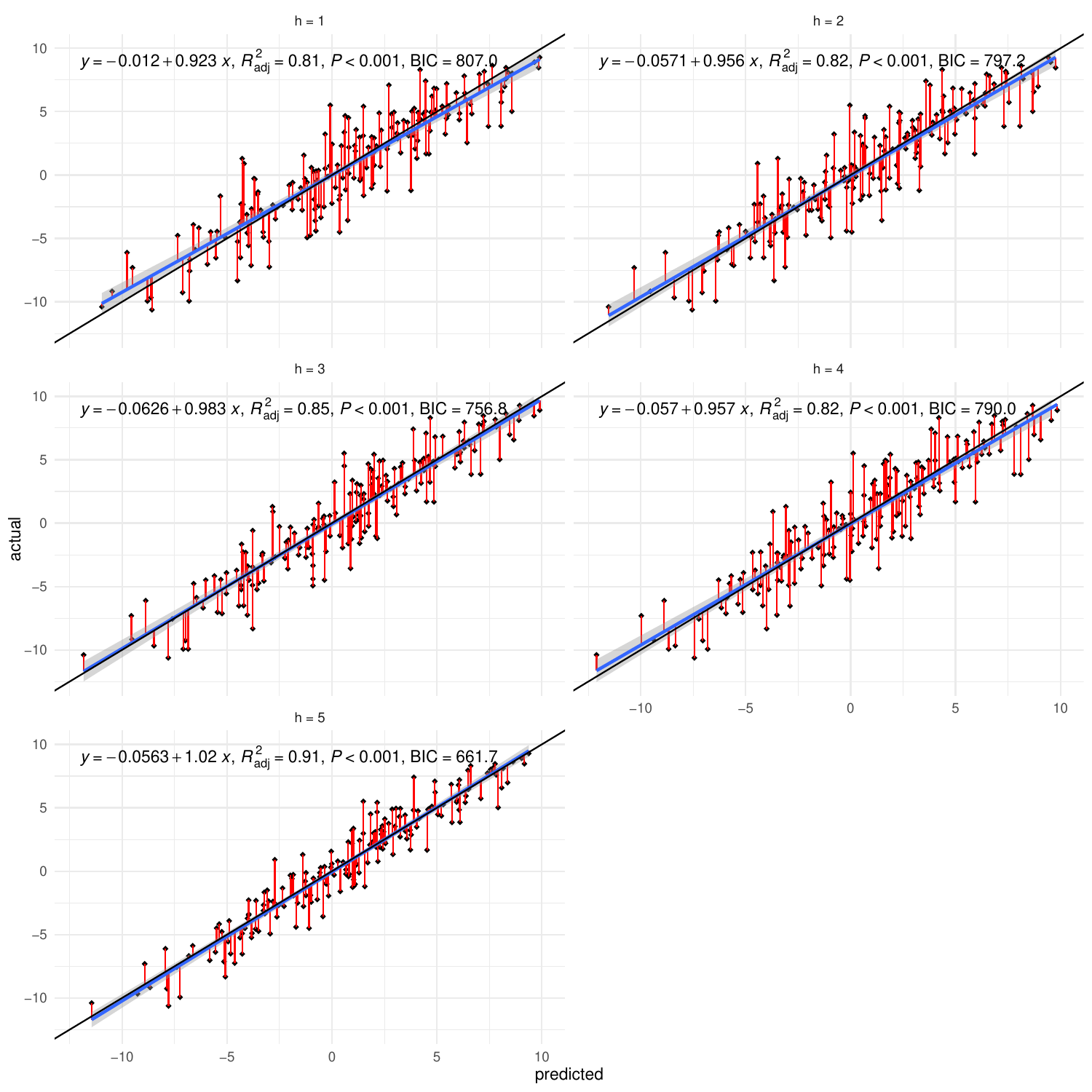}
    \caption{Prediction for different h with $n=1000$, $p=50$, and $q=10$}
    \label{fig:prediction_simulation_accuracy}
\end{figure}

\begin{table}[ht!]
\caption{Prediction accuracy for simulation with $n=1000$, $p=50$, and $q=10$}
\centering
\scriptsize
\begin{tabular}{|r|r|r|r|r|r|}
  \hline
 Metrics & h=1 & h=2 & h=3 & h=4 & h=5 \\ 
  \hline
ME & 0.02 & 0.05 & 0.04 & 0.05 & 0.02 \\ 
  \hline
MAE & 1.67 & 1.60 & 1.47 & 1.64 & 1.11 \\ 
  \hline
MSE & 4.46 & 4.24 & 3.47 & 4.33 & 2.01 \\ 
  \hline
r & 0.91 & 0.91 & 0.93 & 0.91 & 0.96 \\ 
  \hline
R² & 0.83 & 0.83 & 0.86 & 0.83 & 0.92 \\ 
  \hline
\end{tabular}
\label{tab:sim_prediction_accuracy}
\end{table}

\section{Real data application}
\label{real_data_application}
\subsection{Water table depth prediction}
Water table depth prediction is essential for the sustainable management of water resources, agricultural productivity, environmental health, infrastructure development, and hazard management. Monitoring and understanding the depth of the water table helps in predicting and mitigating the effects of climate change, managing water supply, and protecting ecosystems. We have chosen one study site WS77 located at the US Forest Service Santee Experimental Forest/Center for Forested Wetlands in \href{https://www.fs.usda.gov/rds/archive/Catalog/RDS-2019-0033}{wetland online available dataset}. The location diagram is given in the Figure~\ref{fig:map}. A detailed description and processing of the data can be found in \cite{manna2024hydrology}. We have used hydroclimatic data from 2006 to 2019 which consists of water table depth, air temperature, soil temperature, rainfall, water flow, solar radiation, net radiation, Potential EvapoTranspiration (PET), relative humidity (RH), windspeed, wind direction, and vapor pressure in daily scale. There were missing data for different hydroclimatic variables. So we used ``stl'' package in ``R'' to obtain the seasonal and trend component for each of the variables. Then we used linear interpolation for seasonal and trend components using $\text{na\_interpolation}$ from ``imputeTS'' package in R. The rainfall for the same day might not immediately affect the current day's water table depth. So, we have used four previous days lag of the covariates and the water table depth. As hydrologists are very interested in knowing the performance of the prediction of the depth of the water table in recent time period, we predicted the last 400 days in the available dataset and demonstrated in the Figure~\ref{fig:water_table_depth_time_series_no_transformation}. The residual analysis is also plotted in the Figure~\ref{fig:water_table_depth_notransformation} and the results are reflected in Table~\ref{tab:comparizon_table_without_transformation}.

\begin{figure}[ht!]
\centering
    \includegraphics[width=0.7\textwidth]{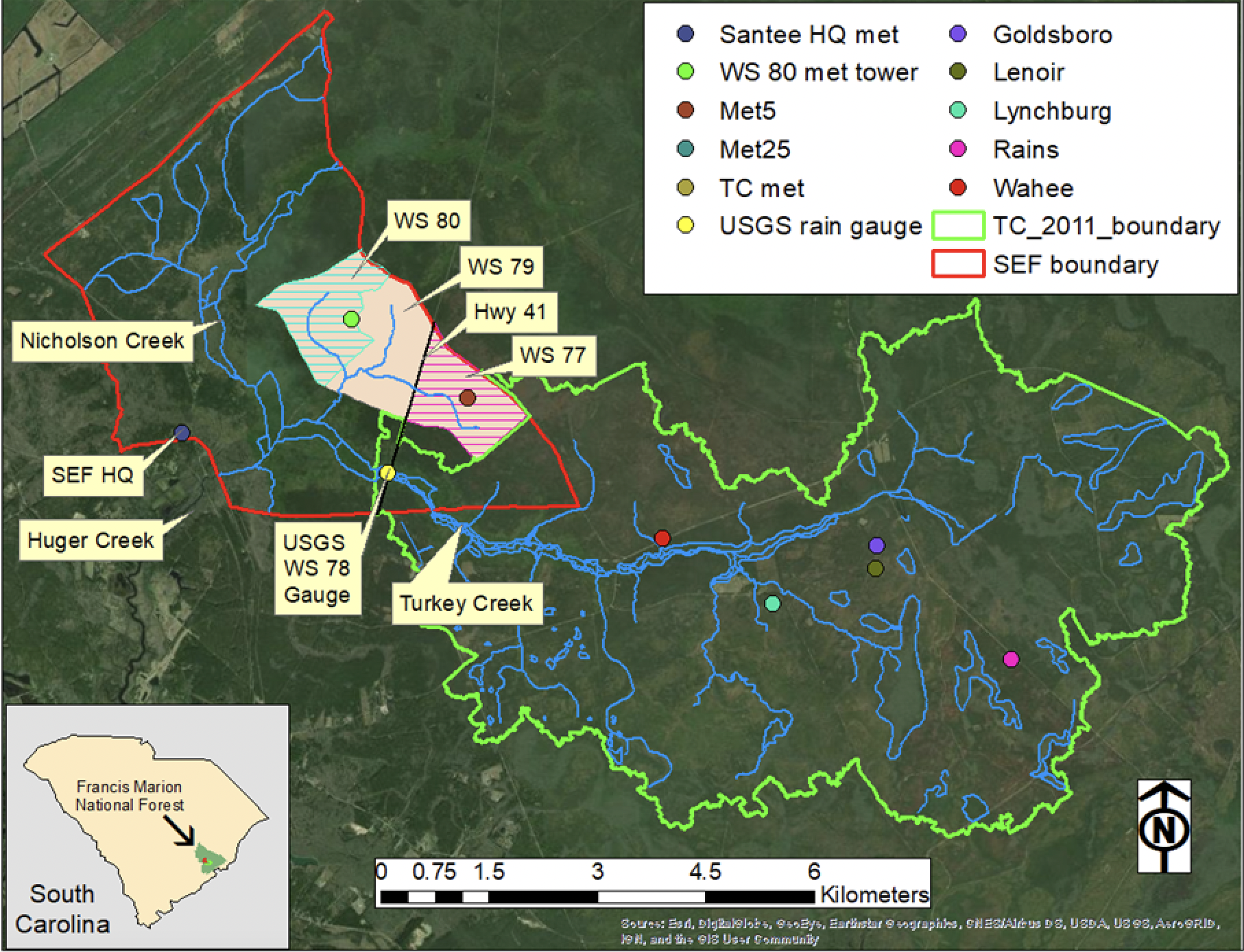}
    \caption{Location map of study watershed WS77, \cite{amatya2022long};\cite{mehansimulation} }
    \label{fig:map}
\end{figure}

\begin{table}[ht!]
\caption{Comparison table without transformation}
\centering
\scriptsize
\begin{tabular}{|l|r|r|r|r|r|}
  \hline
  Metric & h=1 & h=2 & h=3 & h=4 & h=5 \\ 
  \hline
  ME    & -0.41 & -0.45 & -0.52 & -0.72 & -0.49 \\ 
  \hline
  MAE   & 14.92 & 11.93 & 9.84 & 5.16 & 11.17 \\ 
  \hline
  MSE   & 575.79 & 363.55 & 293.82 & 70.95 & 375.55 \\ 
  \hline
  r     & 0.86 & 0.90 & 0.92 & 0.98 & 0.90 \\ 
  \hline
  R²    & 0.66 & 0.79 & 0.83 & 0.96 & 0.78 \\ 
  \hline
\end{tabular}
\label{tab:comparizon_table_without_transformation}
\end{table}

\begin{figure}[ht!]
    \centering
    \includegraphics[width=0.7\textwidth]{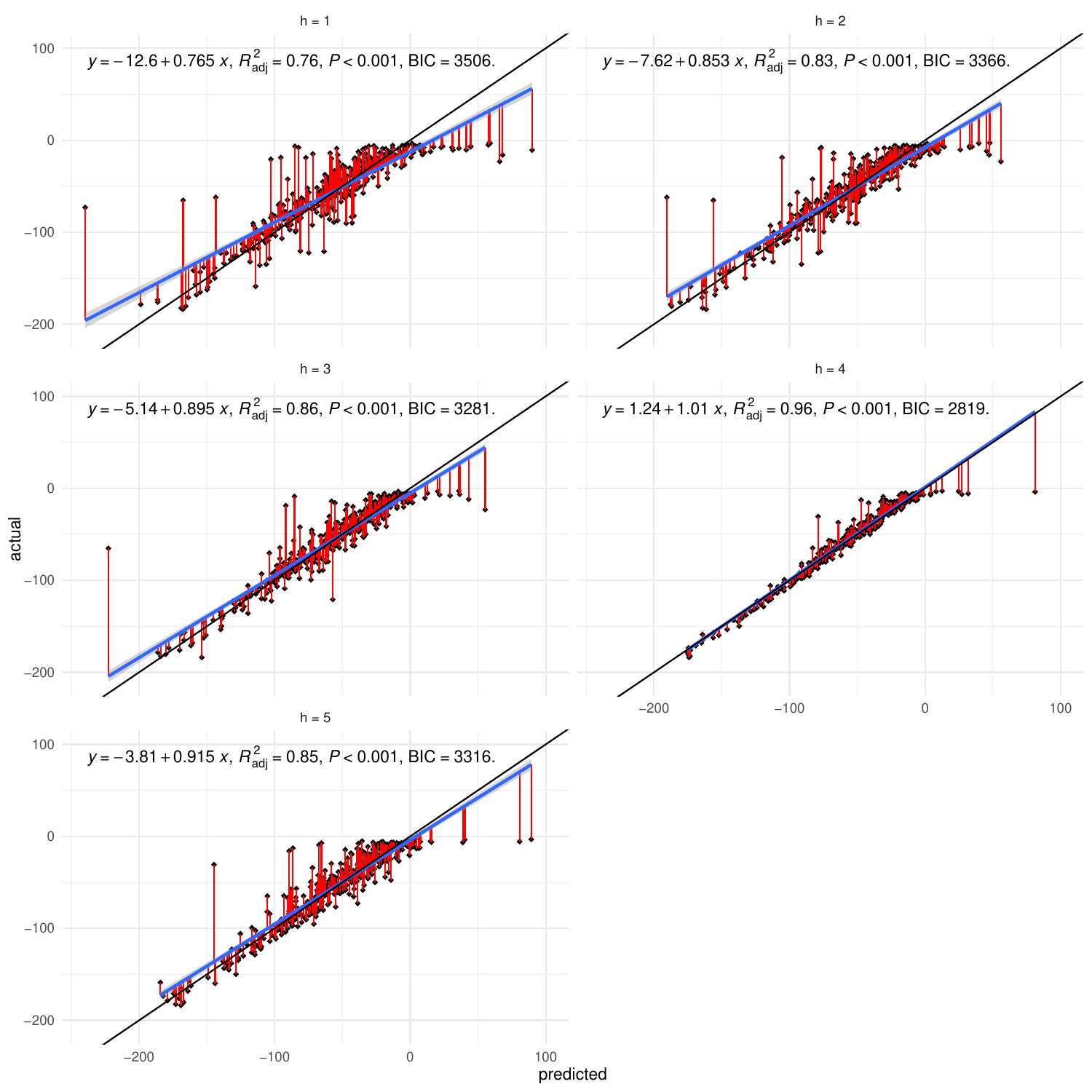}
    \caption{Actual vs predicted comparison for different h without transformation}
    \label{fig:water_table_depth_notransformation}
\end{figure}

\begin{figure}[ht!]
    \centering
    \includegraphics[width=0.7\textwidth]{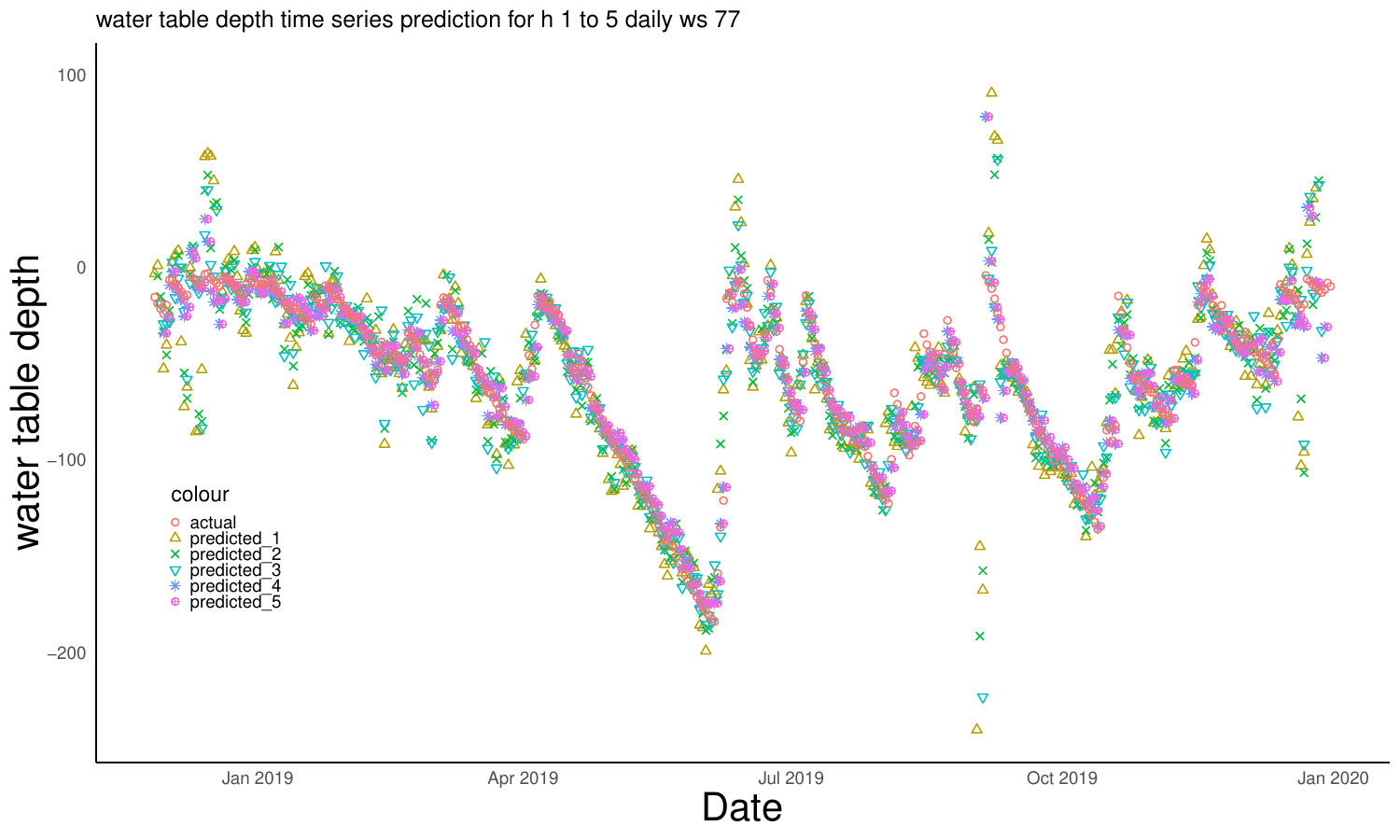}
    \caption{Actual vs predicted comparison for different h without transformation}
    \label{fig:water_table_depth_time_series_no_transformation}
\end{figure}

\begin{figure}[ht!]
    \centering
    \includegraphics[width=0.7\textwidth]{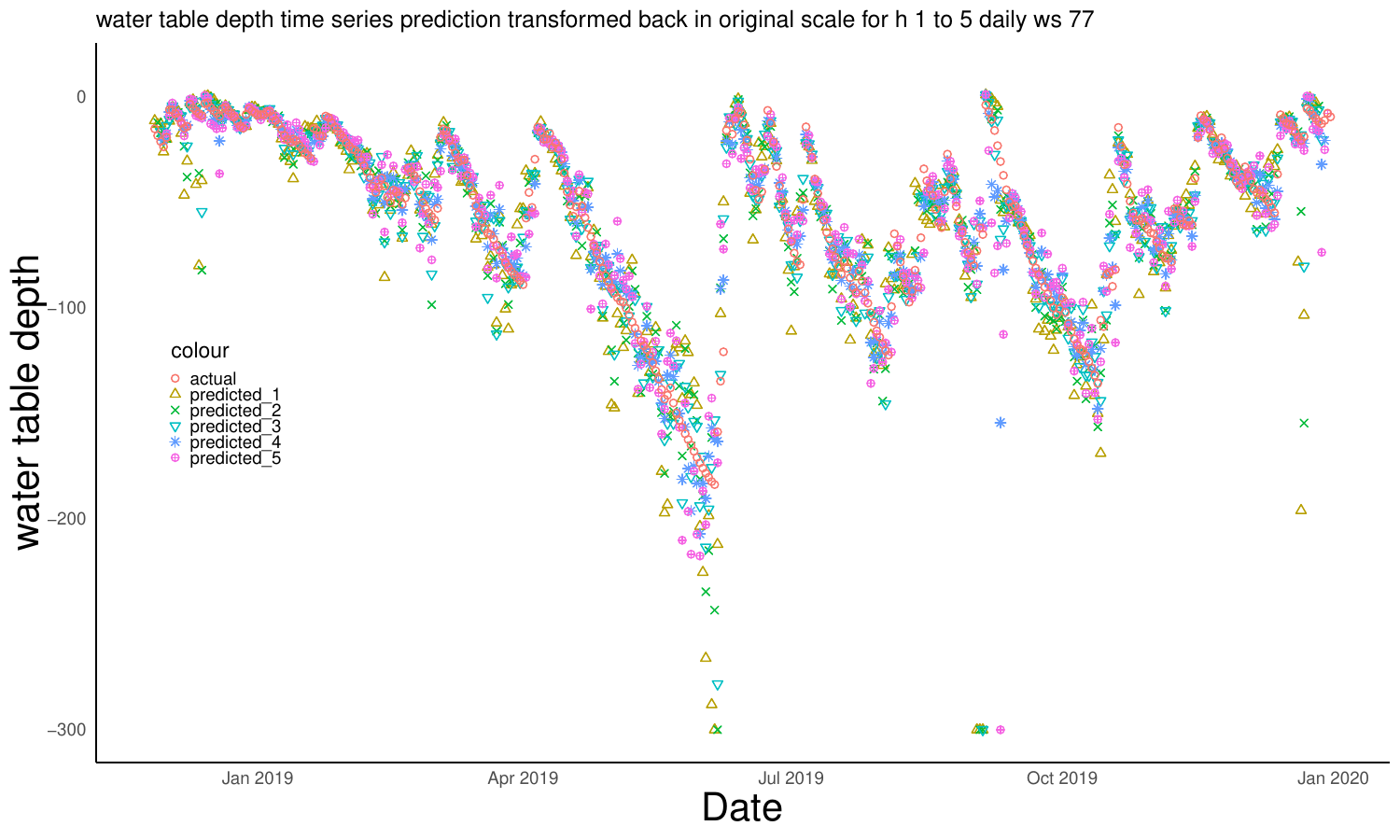}
    \caption{Actual vs predicted comparison for different h after re transformation to original scale}
    \label{fig:water_table_depth_time_series_re_transformation}
\end{figure}

\begin{figure}[ht!]
    \centering
    \includegraphics[width=0.7\textwidth]{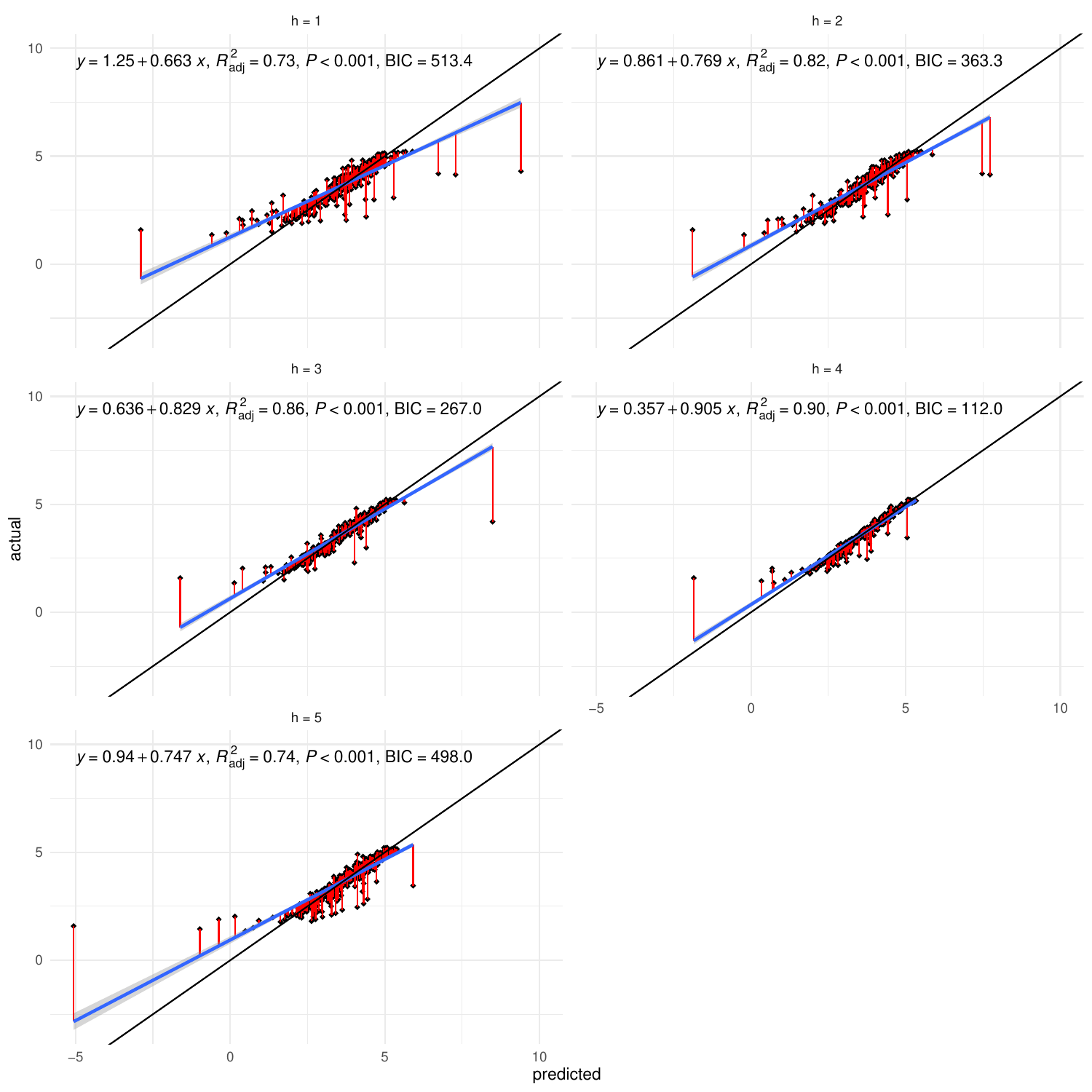}
    \caption{Actual vs predicted comparison for different h after $log(-x+1)$ transformation}
    \label{fig:water_table_depth_l_transformation}
\end{figure}

\begin{figure}[ht!]
    \centering
    \includegraphics[width=0.7\textwidth]{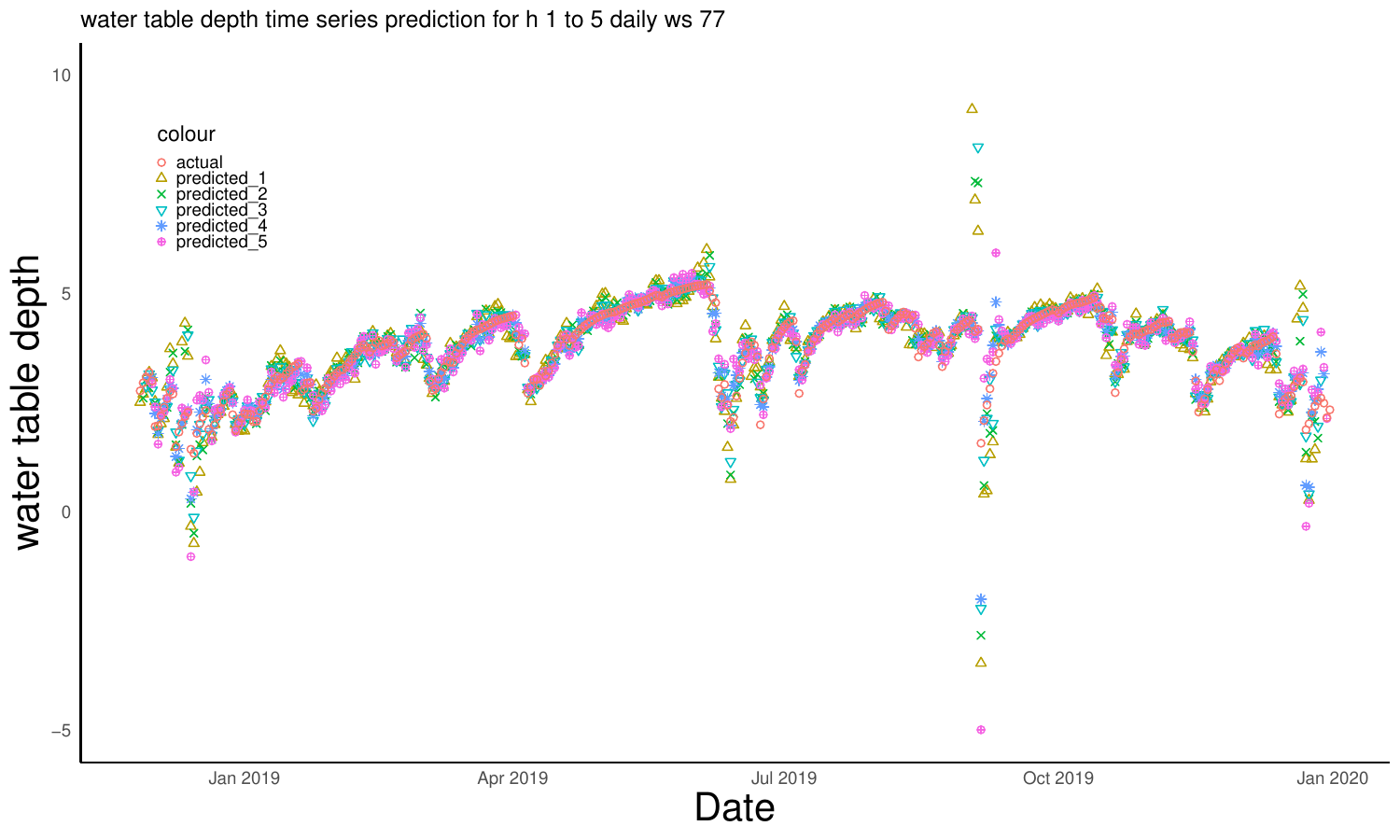}
    \caption{Actual vs predicted comparison for different h with $\log(-y+1)$ transformation}
    \label{fig:water_table_depth_time_series_l_transformation}
\end{figure}

\begin{table}[ht!]
\caption{Comparison table with transformation}
\centering
\scriptsize
\begin{tabular}{|l|r|r|r|r|r|}
  \hline
  Metrics & h=1 & h=2 & h=3 & h=4 & h=5 \\ 
  \hline
  ME    & 0.00 & -0.01 & -0.00 & -0.00 & -0.00 \\ 
  \hline
  MAE   & 0.31 & 0.23 & 0.18 & 0.15 & 0.24 \\ 
  \hline
  MSE   & 0.36 & 0.23 & 0.15 & 0.08 & 0.25 \\ 
  \hline
  r     & 0.85 & 0.89 & 0.92 & 0.95 & 0.87 \\ 
  \hline
  R²    & 0.52 & 0.70 & 0.80 & 0.90 & 0.68 \\ 
  \hline
\end{tabular}
\label{tab:comparizon_table_with_transformation}
\end{table}

\begin{table}[ht!]
\caption{Comparison table with transformation and then re-transformation}
\centering
\scriptsize
\begin{tabular}{|l|r|r|r|r|r|}
  \hline
  Metrics & h=1 & h=2 & h=3 & h=4 & h=5 \\ 
  \hline
  ME    & -36.84 & -10.59 & -12.11 & -0.50 & -1.66 \\ 
  \hline
  MAE   & 46.10 & 18.10 & 18.58 & 6.35 & 11.07 \\ 
  \hline
  MSE   & 346710.56 & 15792.98 & 47679.34 & 98.16 & 455.56 \\ 
  \hline
  r     & 0.10 & 0.34 & 0.21 & 0.98 & 0.89 \\ 
  \hline
\end{tabular}
\label{tab:comparizon_table_with_transformation_retransformation}
\end{table}

\subsection{Stock market data prediction}

\subsubsection{S\&P 500 prediction}
a) We collected the stock market data for several stocks such as GSPC (S\&P 500), AAPL (Apple Inc.), MSFT (Microsoft Corporation), GOOGL (Alphabet Inc.), AMZN (Amazon.com Inc.), TSLA (Tesla Inc.), META (Meta Platforms Inc.), NVDA (NVIDIA Corporation), BRK-B (Berkshire Hathaway Inc. Class B), JNJ (Johnson \& Johnson), V (Visa Inc.), MA (Mastercard Inc.), DIS (The Walt Disney Company), HD (The Home Depot Inc.), NFLX (Netflix Inc.), PFE (Pfizer Inc.), ADBE (Adobe Inc.), INTC (Intel Corporation), CMCSA (Comcast Corporation), CSCO (Cisco Systems Inc.), KO (The Coca-Cola Company), PEP (PepsiCo Inc.), T (AT\&T Inc.), MRK (Merck \& Co. Inc.), WMT (Walmart Inc.), ABT (Abbott Laboratories), XOM (Exxon Mobil Corporation), CVX (Chevron Corporation), MCD (McDonald’s Corporation), NKE (Nike Inc.), UNH (UnitedHealth Group Incorporated), GS (The Goldman Sachs Group Inc.), JPM (JPMorgan Chase \& Co.), MS (Morgan Stanley), BA (The Boeing Company), IBM (International Business Machines Corporation), BMY (Bristol-Myers Squibb Company), ORCL (Oracle Corporation), AMGN (Amgen Inc.), GILD (Gilead Sciences Inc.), TXN (Texas Instruments Incorporated), CVS (CVS Health Corporation), COST (Costco Wholesale Corporation), LMT (Lockheed Martin Corporation), ATVI (Activision Blizzard Inc.), EA (Electronic Arts Inc.), WFC (Wells Fargo \& Company), UBER (Uber Technologies Inc.), RBLX (Roblox Corporation), and SHOP (Shopify Inc.) from 2007 and predict S\&P 500 stock value for the last 200 days with our model. Figure \ref{fig:snp500_prediction_200} and table \ref{tab:snp500_200} represent the residuals of the S\&P prediction. We present NRMSE which is defined as \[\text{NRMSE} = \frac{\sqrt{\frac{1}{N-n} \sum_{i=n+1}^{N} (y_i - \hat{y}_i)^2}}{y_{\text{max}} - y_{\text{min}}}.
\] The MSE was very small, so we preferred to present NRMSE in the Table~\ref{tab:snp500_200}.  Table \ref{tab:inc_prob_snp500_200} keeps the included variables which helps to predict the S\&P stock value.
\begin{figure}[h!]
    \centering
    \includegraphics[width=0.7\textwidth]{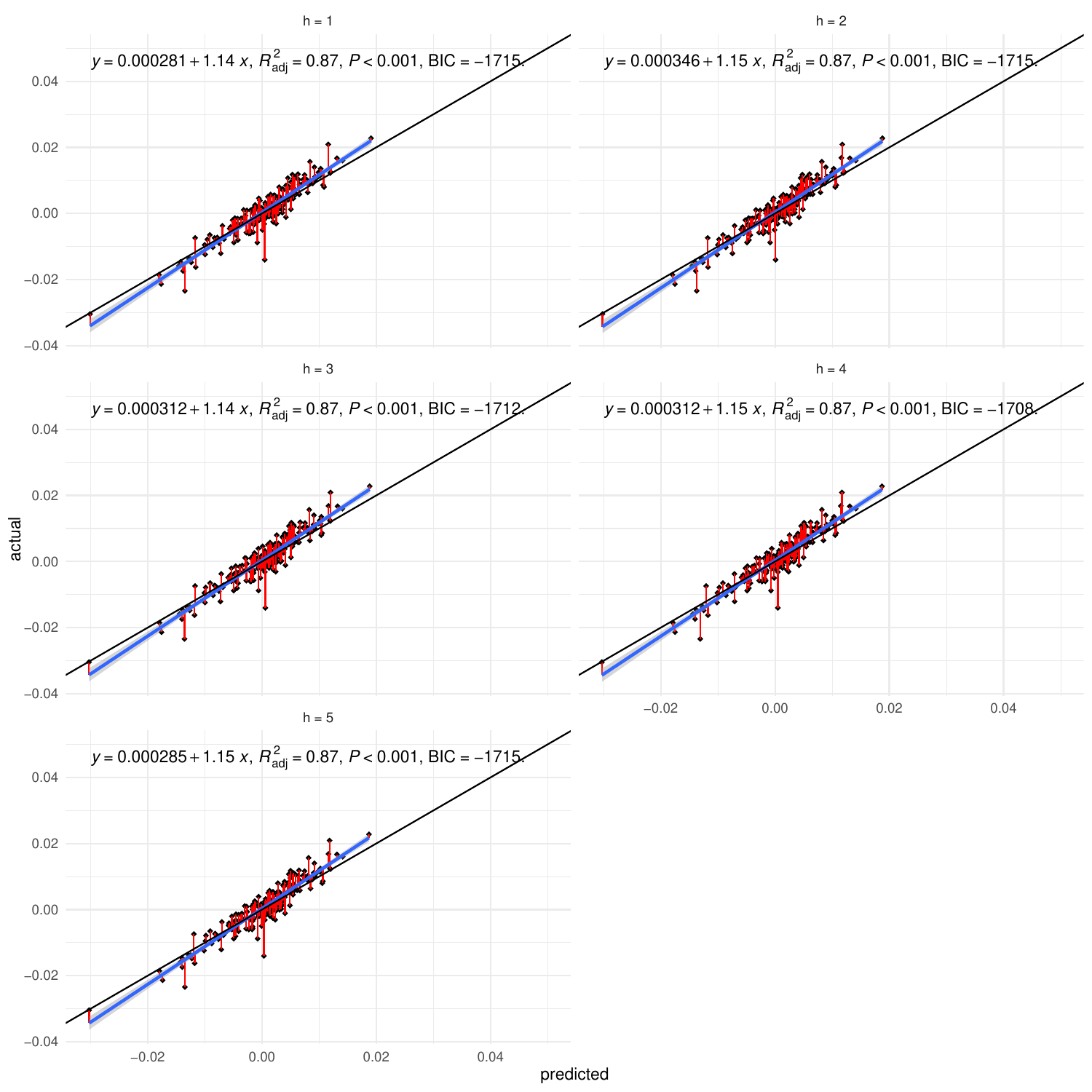}
    \caption{Actual vs predicted comparison for different h for S\&P 500 stock with cutoff $.001\times \frac{1}{p}$ for variable selection}
    \label{fig:snp500_prediction_200}
\end{figure}

\begin{table}[ht!]
\caption{Inclusion probability for S\&P 500 stock prediction for the last 200 days based on stable stocks from NASDAQ (training data from 2007)}
\centering
\scriptsize
\begin{tabular}{|c|p{10cm}|}
\hline
  \textbf{Inclusion Probability} & \textbf{Selected Variables} \\ 
  \hline
  1 \multirow{25}{*}{} & BA, CVX, HD, MA, WFC, AMGN, UNH, CMCSA, INTC, ADBE, PFE, GOOGL, NKE, AMZN, MSFT, KO, TXN, PEP, MS, XOM, T, BRK.B, AAPL, ABT, DIS \\ 
  \hline
  0.99 & ORCL \\ 
  \hline
  0.97 & IBM \\ 
  \hline
  0.95 & GILD \\ 
  \hline
  0.93 & NFLX, CSCO \\ 
  \hline
  0.90 & CVS, LMT \\ 
  \hline
  0.85 & JPM \\ 
  \hline
  0.74 & BMY \\ 
  \hline
  0.19 & GS \\ 
  \hline
  0.14 & EA \\ 
  \hline
  0.09 & META \\ 
  \hline
  0.08 & NVDA \\ 
  \hline
  0.04 & TSLA \\ 
  \hline
  0.02 & MRK, MCD \\ 
  \hline
  0.01 & V, JNJ, WMT \\ 
  \hline
\end{tabular}
\label{tab:inc_prob_snp500_200}
\end{table}

\begin{table}[ht!]
\caption{Comparison for S\&P 500 Stock Prediction (Last 200 Days) Based on Stable NASDAQ Stocks (Data from 2007)}
\centering
\scriptsize
\begin{tabular}{|c|c|c|c|c|c|}
  \hline
  \textbf{Metric} & \textbf{h=1} & \textbf{h=2} & \textbf{h=3} & \textbf{h=4} & \textbf{h=5} \\ 
  \hline
  NRMSE (\%) & 38.00 & 37.90 & 38.10 & 38.40 & 38.00 \\ 
  \hline
  \(r\) & 0.93 & 0.93 & 0.93 & 0.93 & 0.93 \\ 
  \hline
  \(R^2\) & 0.85 & 0.86 & 0.85 & 0.85 & 0.85 \\ 
  \hline
\end{tabular}
\label{tab:snp500_200}
\end{table}

b) Another analysis predicts the S\&P stock values for the last 2 months based on data from 2021, as stock market data older than 2021 may not significantly influence the current pattern of stock values. Figure \ref{fig:snp500_prediction_60} and table \ref{tab:snp500_60} present the residual analysis for predicting S\&P stock values for the last 2 months. Table \ref{tab:inc_prob_snp500_60} shows the inclusion probabilities of important stock values for predicting the current S\&P stock values.
\begin{figure}[h!]
    \centering
    \includegraphics[width=0.7\textwidth]{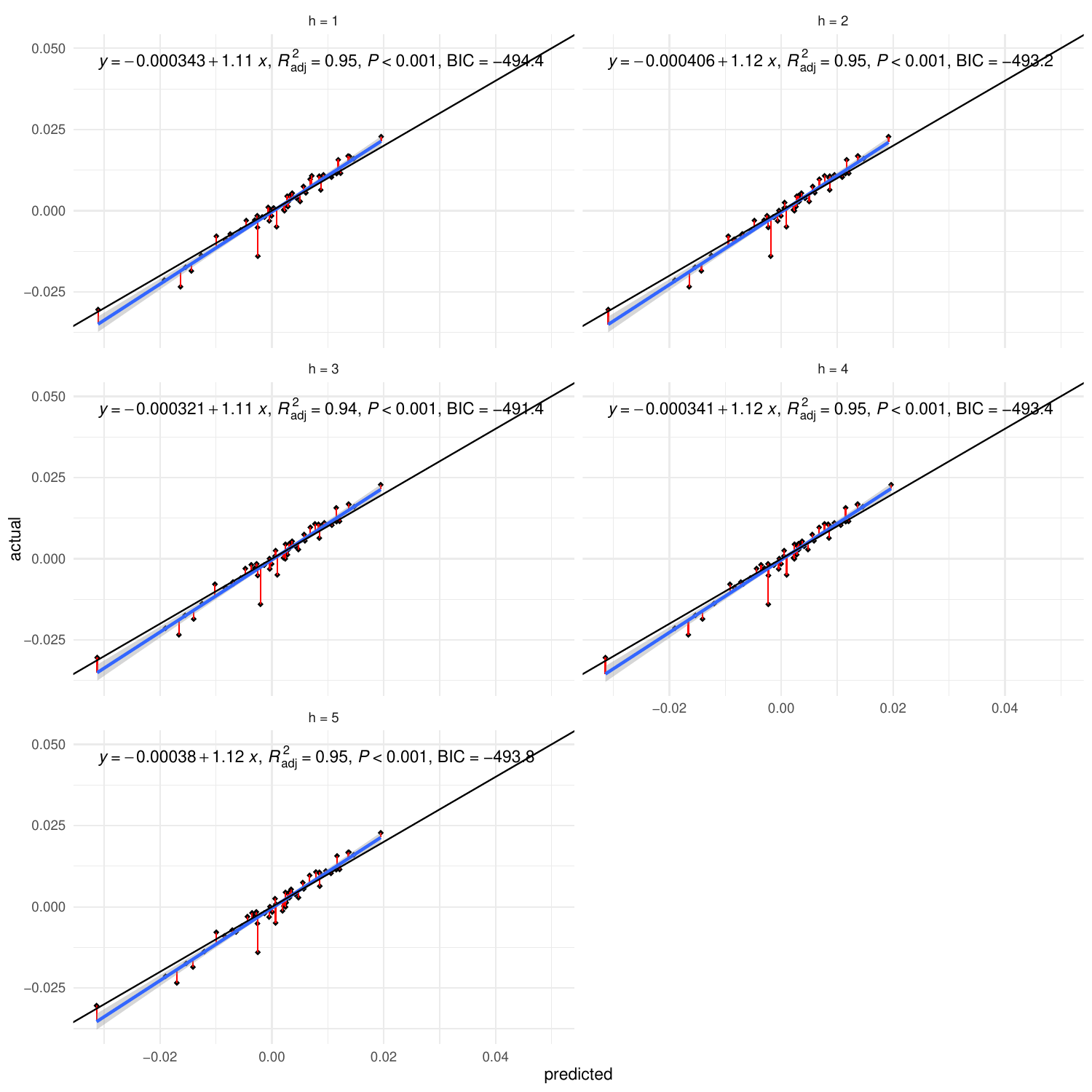}
    \caption{Actual vs predicted comparison for different h for S\&P 500 stock with cutoff $.001\times \frac{1}{p}$ for variable selection using the data from 2021-02-11}
    \label{fig:snp500_prediction_60}
\end{figure}

\begin{table}[ht!]
\caption{Inclusion probability for S\&P 500 stock prediction for the last 60 days based on stable stocks from NASDAQ (training data from 2021)}
\centering
\scriptsize
\begin{tabular}{|c|p{8cm}|}
  \hline
  \textbf{Inclusion Probability} & \textbf{Selected Variables} \\ 
  \hline
  1 \multirow{14}{*}{} & BA, HD, MSFT, TXN, BRK.B, NVDA, AAPL, AMZN, ABT, DIS, META, UNH, TSLA, MA \\ 
  \hline
  0.99 & ADBE \\ 
  \hline
  0.97 & GS \\ 
  \hline
  0.96 & CVX, GOOGL, NKE \\ 
  \hline
  0.93 & JPM \\ 
  \hline
  0.90 & BMY \\ 
  \hline
  0.83 & MCD \\ 
  \hline
  0.75 & KO \\ 
  \hline
  0.51 & COST \\ 
  \hline
  0.25 & PEP \\ 
  \hline
  0.20 & T \\ 
  \hline
  0.14 & CSCO \\ 
  \hline
\end{tabular}
\label{tab:inc_prob_snp500_60}
\end{table}

\begin{table}[ht!]
\caption{Comparison for S\&P 500 Stock Prediction (Last 60 Days) Based on Stable NASDAQ Stocks}
\centering
\scriptsize
\begin{tabular}{|c|c|c|c|c|c|}
  \hline
  \textbf{Metric} & \textbf{h=1} & \textbf{h=2} & \textbf{h=3} & \textbf{h=4} & \textbf{h=5} \\ 
  \hline
  NRMSE (\%) & 24.60 & 25.10 & 25.10 & 25.10 & 25.10 \\ 
  \hline
  \(r\) & 0.97 & 0.97 & 0.97 & 0.97 & 0.97 \\ 
  \hline
  \(R^2\) & 0.94 & 0.94 & 0.94 & 0.94 & 0.94 \\ 
  \hline
\end{tabular}
\label{tab:snp500_60}
\end{table}

\section{Discussion}
\label{Discussion}
For theoretical analysis, one may write a similar model for AR(1),
\[Y-aY_{1}=\mathbf{X} \betavec + \boldsymbol{\epsilon},\]
where $Y_{1}=\left(0,y_{1},y_{2},\cdots,y_{n}\right)^{T}.$
Now, if we integrate out \( a \) and \( \boldsymbol{\beta} \) with a normal independent prior \( \pi(a) \sim \mathcal{N}(0, \sigma^2 \tau_{n}^2) \) and \( p(\beta_{[-j]} | \sigma^2) \propto \mathcal{N}(0, \sigma^2 \tau^2_n I) \) respectively, the likelihood \( p(Y | \beta_{j}, \sigma^2) \) might not be a normal density. For higher order $AR(p)$, the conditions might become more complex and harder to interpret. So we focus on the simple and interpretable model \ref{final_model}.

In the Figure~\ref{fig:water_table_depth_time_series_no_transformation}, one might observe that water table depth sometimes becomes positive when there is a huge rainfall in the coastal areas. To alleviate this problem we have used $\log(-y+1)$ transformation for the water table depth demonstrated in the figure 
\ref{fig:water_table_depth_time_series_l_transformation} and figure \ref{fig:water_table_depth_l_transformation}. The residual analysis is done in the Table~\ref{tab:comparizon_table_with_transformation}. After doing re-transformation of the water table depth to the original scale, we can see the plot in \ref{fig:water_table_depth_time_series_re_transformation} and the residual matrix in the Table~\ref{tab:comparizon_table_without_transformation}. If we observe the water table significantly lower, say <-150 cm, in that case, we observe more variability from the original water table depth to the original scale prediction demonstrated in the figures \ref{fig:water_table_depth_time_series_re_transformation} and \ref{fig:water_table_depth_time_series_no_transformation}. In the tables \ref{tab:inclusion_prob_no_transformation_wtd} and \ref{tab:inclusion_prob_with_transformation_wtd} we represent the inclusion probabilities of different variables and their lag which the model uses for prediction. The included variables are almost consistent in both of the cases. For the model without transformation, three lagged values of the water table depth, and with the transformation, the model is likely to pick a few more extra lagged values of water table depth.
\begin{table}[h!]
\centering
\scriptsize
\caption{Inclusion probabilities for water table depth without transformation}
\begin{tabular}{|c|p{8cm}|}
\hline
\textbf{Inclusion Probability} & \textbf{Variables} \\ 
\hline
1.0000 & DailyFlow \\ 
\hline
0.9994 & Rainfall\_2 \\ 
\hline
0.9964 & Rainfall\_3 \\ 
\hline
0.9830 & Windspeed\_\_4 \\ 
\hline
0.9596 & Air\_Temp\_C\_4 \\ 
\hline
0.7412 & RH\_4 \\ 
\hline
0.7308 & Vapor\_Kpa\_4 \\ 
\hline
0.5838 & RH\_1 \\ 
\hline
0.5606 & Windspeed\_\_1 \\ 
\hline
0.3962 & Windspeed\_ \\ 
\hline
0.3932 & RH \\ 
\hline
0.3672 & Air\_Temp\_C \\ 
\hline
0.2896 & PET\_4 \\ 
\hline
0.2878 & PET \\ 
\hline
0.2422 & Rainfall\_4 \\ 
\hline
0.1092 & Air\_Temp\_C\_1 \\ 
\hline
0.0620 & Soil\_Temp\_\_3 \\ 
\hline
0.0500 & Soil\_Temp\_ \\ 
\hline
0.0386 & Rainfall\_1 \\ 
\hline
0.0308 & Soil\_Temp\_\_4 \\ 
\hline
0.0244 & Soil\_Temp\_\_1 \\ 
\hline
\end{tabular}
\label{tab:inclusion_prob_no_transformation_wtd}
\end{table}

\begin{table}[h!]
\centering
\scriptsize
\caption{Inclusion probabilities for water table depth with transformation}
\begin{tabular}{|c|p{4cm}|}
\hline
\textbf{Inclusion Probability} & \textbf{Variables} \\ 
\hline
1\multirow{5}{*}{} & Rainfall\_4, Rainfall\_3, Rainfall\_2, Rainfall\_1, DailyFlow \\ 
\hline
0.9998 & Windspeed\_\_4 \\ 
\hline
0.9614 & Air\_Temp\_C\_4 \\ 
\hline
0.9594 & Vapor\_Kpa\_4 \\ 
\hline
0.6618 & RH \\ 
\hline
0.5972 & Windspeed\_ \\ 
\hline
0.5468 & RH\_4 \\ 
\hline
0.4972 & RH\_2 \\ 
\hline
0.4864 & Windspeed\_\_2 \\ 
\hline
0.4748 & PET\_4 \\ 
\hline
0.4444 & RH\_1 \\ 
\hline
0.4368 & Windspeed\_\_1 \\ 
\hline
0.3684 & Soil\_Temp\_\_1 \\ 
\hline
0.3104 & Soil\_Temp\_ \\ 
\hline
0.1632 & Soil\_Temp\_\_2 \\ 
\hline
0.1510 & Soil\_Temp\_\_3 \\ 
\hline
0.0508 & PET\_2 \\ 
\hline
0.0354 & Air\_Temp\_C \\ 
\hline
0.0334 & PET \\ 
\hline
0.0206 & RH\_3 \\ 
\hline
\end{tabular}
\label{tab:inclusion_prob_with_transformation_wtd}
\end{table}

In this comparative study of inclusion probabilities for water table depth, we evaluate the influence of variables with and without a transformation applied to the dataset. Without transformation, \textbf{DailyFlow} has the highest inclusion probability of 1.0000, closely followed by rainfall-related variables (\textbf{Rainfall\_2}, \textbf{Rainfall\_3}) and \textbf{Windspeed\_\_4}, all having probabilities above 0.98. Variables such as \textbf{RH\_4} and \textbf{Vapor\_Kpa\_4} also have relatively high inclusion probabilities, indicating their relevance to the model. However, some variables, like \textbf{Soil\_Temp\_\_1} and \textbf{Rainfall\_1}, exhibit very low inclusion probabilities, suggesting lesser importance.

In contrast, with the transformation applied, \textbf{Rainfall} variables (\textbf{Rainfall\_4}, \textbf{Rainfall\_3}, \textbf{Rainfall\_2}, \textbf{Rainfall\_1}) dominate the model with perfect inclusion probabilities (1.0000), alongside \textbf{DailyFlow}. The influence of \textbf{Windspeed\_\_4} remains high, while other variables such as \textbf{RH} and \textbf{Vapor\_Kpa\_4} continue to exhibit significant relevance. Notably, \textbf{Soil\_Temp} variables increase in importance with transformation, particularly \textbf{Soil\_Temp\_\_1}, which sees a marked rise from 0.0244 to 0.3684. Overall, the transformation enhances the importance of precipitation-related variables and certain temperature and soil metrics, while reducing the inclusion probability for variables such as \textbf{PET} and \textbf{Air\_Temp\_C}.

This comparison highlights the shift in variable importance when a transformation is applied, making certain features like rainfall and soil temperature more critical for predicting water table depth, while others lose their influence.

In the prediction for S\&P 500, the majority of the autocorrelations for S\&P 500 are already captured by the stocks which are themselves correlated and are linearly related to S\&P 500 and so in the residuals we could not find any significant autocorrelations after fitting the model. Both periods (2007 and 2021) highlight certain stable stocks that maintain high inclusion probabilities (close to or equal to 1) consistently across the timeframes. Stocks like BA (Boeing), HD (Home Depot), MSFT (Microsoft), BRK.B (Berkshire Hathaway), AAPL (Apple), AMZN (Amazon), and UNH (UnitedHealth) consistently have inclusion probabilities of 1 across both timeframes. These are likely key contributors to the S\&P 500's performance. Some stocks, such as NVDA (NVIDIA), which has an inclusion probability of 1 in 2021, were absent from the 2007 table, reflecting its growing significance in the market, likely driven by its increasing dominance in the tech industry. META (formerly Facebook) appears with a perfect inclusion probability in 2021 but was absent or had a much lower significance in 2007, aligning with its later market expansion. Certain stocks show a notable decline in their inclusion probability in 2021 compared to 2007. For example, PEP (Pepsi), which had a perfect inclusion probability in 2007, only had a probability of 0.25 in 2021. This might suggest reduced importance in the S\&P 500 prediction over time or less stability in the dataset. T (AT\&T) also shows a significant drop in inclusion probability from 1 in 2007 to 0.20 in 2021. Stocks like NVDA, TSLA (Tesla), and META that were not present or less significant from larger historical data from 2007 have emerged as major players with a perfect inclusion probability in the recent data from 2021. Stocks such as IBM, DIS (Disney), and GS (Goldman Sachs) remain consistently included in the model but show slight differences in inclusion probabilities between the two periods, possibly due to market conditions or sectoral shifts. The two tables highlight how the stock market landscape and key contributors to the S\&P 500 have evolved over time. While traditional large-cap companies (like BA, AAPL, and MSFT) maintain high inclusion probabilities, newer entrants like NVDA, TSLA, and META now dominate, reflecting broader technological and economic shifts. Conversely, older stalwarts like AT\&T and Pepsi have seen a reduction in their market influence in recent years. In our analysis, we did not obtain any significant lag in the stocks and the error process, indicating that the old stocks log return values are not influential in predicting the current S\&P, but with a smaller portfolio, one can invest in the constituent stocks in proportion to table \ref{fig:snp500_prediction_200} to obtain similar returns for S\&P.

\section{Acknowledgement}
\label{Acknowledgement} In this project, we would like to thank \href{https://www.sccoastalinfo.org/shared-contact/devendra-amatya-usda-fs-center-for-forested-wetlands-research/}{Devendra Amatya – USDA FS} and his team for maintaining water table depth data for the U.S. Department of Agriculture. We would also like to thank Dipak Dey, Board of Trustees Distinguished Professor in the University of Connecticut for his necessary suggestions and guidance.

\bibliographystyle{plainnat}
\bibliography{nsf}

\begin{thebibliography}{20}
\providecommand{\natexlab}[1]{#1}
\providecommand{\url}[1]{\texttt{#1}}
\expandafter\ifx\csname urlstyle\endcsname\relax
  \providecommand{\doi}[1]{doi: #1}\else
  \providecommand{\doi}{doi: \begingroup \urlstyle{rm}\Url}\fi

\bibitem[Amatya et~al.(2022)Amatya, Trettin, Harrison, and Arnold]{amatya2022long}
Devendra~M Amatya, Carl~C Trettin, Charles~A Harrison, and Julie~A Arnold.
\newblock Long-term hydro-meteorology and water quality data from low-gradient catchments of varying scales on the santee experimental forest, south carolina.
\newblock \emph{Hydrological Processes}, 36\penalty0 (5):\penalty0 e14549, 2022.

\bibitem[Box et~al.(2015)Box, Jenkins, Reinsel, and Ljung]{box2015time}
George~EP Box, Gwilym~M Jenkins, Gregory~C Reinsel, and Greta~M Ljung.
\newblock \emph{Time series analysis: forecasting and control}.
\newblock John Wiley \& Sons, 2015.

\bibitem[Carvalho et~al.(2009)Carvalho, Polson, and Scott]{carvalho2009handling}
Carlos~M Carvalho, Nicholas~G Polson, and James~G Scott.
\newblock Handling sparsity via the horseshoe.
\newblock In \emph{Artificial intelligence and statistics}, pages 73--80. PMLR, 2009.

\bibitem[Chen and Walker(2019)]{chen2019fast}
Su~Chen and Stephen~G. Walker.
\newblock Fast bayesian variable selection for high dimensional linear models: Marginal solo spike and slab priors.
\newblock \emph{Electronic Journal of Statistics}, 13\penalty0 (1):\penalty0 284--309, 2019.
\newblock \doi{10.1214/18-EJS1529}.
\newblock URL \url{https://doi.org/10.1214/18-EJS1529}.

\bibitem[Dawn et~al.(2021)Dawn, Roy, Manna, and Ghosh]{dawn2021some}
Trisha Dawn, Angshuman Roy, Alokesh Manna, and Anil~K Ghosh.
\newblock Some clustering-based change-point detection methods applicable to high dimension, low sample size data.
\newblock \emph{arXiv preprint arXiv:2111.14012}, 2021.

\bibitem[Dawn et~al.(2025)Dawn, Roy, Manna, and Ghosh]{dawn2025some}
Trisha Dawn, Angshuman Roy, Alokesh Manna, and Anil~K Ghosh.
\newblock Some clustering-based change-point detection methods applicable to high dimension, low sample size data.
\newblock \emph{Journal of Statistical Planning and Inference}, 234:\penalty0 106212, 2025.

\bibitem[Friedman et~al.(2021)Friedman, Hastie, Tibshirani, Narasimhan, Tay, Simon, and Qian]{friedman2021package}
Jerome Friedman, Trevor Hastie, Rob Tibshirani, Balasubramanian Narasimhan, Kenneth Tay, Noah Simon, and Junyang Qian.
\newblock Package ‘glmnet’.
\newblock \emph{CRAN R Repositary}, 595, 2021.
\newblock URL \url{https://cran.r-project.org/package=glmnet}.

\bibitem[Ghosh et~al.(2021)Ghosh, Khare, and Michailidis]{ghosh2021strong}
Satyajit Ghosh, Kshitij Khare, and George Michailidis.
\newblock Strong selection consistency of bayesian vector autoregressive models based on a pseudo-likelihood approach.
\newblock \emph{The Annals of Statistics}, 49\penalty0 (3):\penalty0 1267--1299, 2021.

\bibitem[Goh and Dey(2018)]{goh2018bayesian}
Gyuhyeong Goh and Dipak~K Dey.
\newblock Bayesian map estimation using gaussian and diffused-gamma prior.
\newblock \emph{Canadian Journal of Statistics}, 46\penalty0 (3):\penalty0 399--415, 2018.

\bibitem[Ishwaran and Rao(2005)]{ishwaran2005spike}
Hemant Ishwaran and J~Sunil Rao.
\newblock Spike and slab variable selection: frequentist and bayesian strategies.
\newblock \emph{The Annals of Statistics 33(2) 730–773.}, 2005.

\bibitem[Manna et~al.(2024{\natexlab{a}})Manna, Huang, Dey, Gu, and He]{manna2024interval}
Alokesh Manna, Zijian Huang, Dipak~K Dey, Yuwen Gu, and Robin He.
\newblock Interval estimation of coefficients in penalized regression models of insurance data.
\newblock \emph{arXiv preprint arXiv:2410.01008}, 2024{\natexlab{a}}.
\newblock URL \url{https://arxiv.org/abs/2410.01008}.

\bibitem[Manna et~al.(2024{\natexlab{b}})Manna, Mehan, and Amatya]{manna2024hydrology}
Alokesh Manna, Sushant Mehan, and Devendra~M. Amatya.
\newblock Development of a statistical predictive model for daily water table depth and important variables selection for inference.
\newblock \emph{arXiv preprint arXiv:2410.01001}, 2024{\natexlab{b}}.
\newblock URL \url{https://arxiv.org/abs/2410.01001}.

\bibitem[Manna et~al.(2025)Manna, Sett, Dey, Gu, Schifano, and He]{manna2025distributionfreeinferencelightgbmglm}
Alokesh Manna, Aditya~Vikram Sett, Dipak~K. Dey, Yuwen Gu, Elizabeth~D. Schifano, and Jichao He.
\newblock Distribution-free inference for lightgbm and glm with tweedie loss.
\newblock \emph{arXiv preprint arXiv:2507.06921}, 2025.
\newblock URL \url{https://arxiv.org/abs/2507.06921}.

\bibitem[Mehan et~al.(2025)Mehan, Manna, and Amatya]{mehansimulation}
Sushant Mehan, Alokesh Manna, and Devendra~Man Amatya.
\newblock Simulation of water table depths in a forested wetland system in the carolinas using data-driven approaches.
\newblock \emph{AGU25}, 2025.

\bibitem[Nicholson et~al.(2017)Nicholson, Matteson, and Bien]{nicholson2017bigvar}
William Nicholson, David Matteson, and Jacob Bien.
\newblock Bigvar: Tools for modeling sparse high-dimensional multivariate time series.
\newblock \emph{arXiv preprint arXiv:1702.07094}, 2017.

\bibitem[Ocampo and Rodr{\'\i}guez(2012)]{ocampo2012introductory}
Sergio Ocampo and Norberto Rodr{\'\i}guez.
\newblock An introductory review of a structural var-x estimation and applications.
\newblock \emph{Revista Colombiana de Estad{\'\i}stica}, 35\penalty0 (3):\penalty0 479--508, 2012.

\bibitem[Scott and Scott(2023)]{scott2023package}
Steven~L Scott and Maintainer Steven~L Scott.
\newblock Package ‘boomspikeslab’.
\newblock \emph{CRAN R Repositary}, 2023.
\newblock URL \url{https://cran.r-project.org/package=BoomSpikeSlab}.

\bibitem[Tibshirani(1996)]{tibshirani1996regression}
Robert Tibshirani.
\newblock Regression shrinkage and selection via the lasso.
\newblock \emph{Journal of the Royal Statistical Society Series B: Statistical Methodology}, 58\penalty0 (1):\penalty0 267--288, 1996.

\bibitem[Tua{\c{c}} and Arslan(2021)]{tuacc2021variable}
Yetkin Tua{\c{c}} and Olcay Arslan.
\newblock Variable selection in regression model with ar (p) error terms based on heavy tailed distributions.
\newblock \emph{arXiv preprint arXiv:2108.13755}, 2021.

\bibitem[Wang et~al.(2007)Wang, Li, and Tsai]{wang2007regression}
Hansheng Wang, Guodong Li, and Chih-Ling Tsai.
\newblock Regression coefficient and autoregressive order shrinkage and selection via the lasso.
\newblock \emph{Journal of the Royal Statistical Society Series B: Statistical Methodology}, 69\penalty0 (1):\penalty0 63--78, 2007.

\end{thebibliography}
\newpage

\section{Appendix: Proof of Technical Results}
\label{appendix}
We present proofs of the technical results in this supplementary document.

\noindent {\bf Proof of Lemma \ref{lem:psd_matrix}}:\\
\begin{proof}
\begin{align*}
    &(I - \mathbf{\tilde{H}}_j)\\
    &=\left(\tau^2_n \mathbf{X}_{[-j]}\mathbf{X}_{[-j]}^{T}+(\mathbf{A}(\phivec)^{T}\mathbf{A}(\phivec))^{-1}\right)^{-1}\\
    &=\left(\tau^2_n \mathbf{A}^{-1}(\phivec) \mathbf{A}(\phivec) \mathbf{X}_{[-j]}\mathbf{X}_{[-j]}^{T} \mathbf{A}^{T}(\phivec)  (\mathbf{A}^{T}(\phivec))^{-1}+\mathbf{A}^{-1}(\phivec) (\mathbf{A}^{T}(\phivec))^{-1}\right)^{-1}\\
    &=\left( \mathbf{A}^{-1}(\phivec)[ \tau^2_n \mathbf{A}(\phivec) \mathbf{X}_{[-j]}\mathbf{X}_{[-j]}^{T} \mathbf{A}^{T}(\phivec) + I]^{-1} (\mathbf{A}^{T}(\phivec))^{-1}\right)^{-1}\\
    &=\mathbf{A}^{T}(\phivec) [\tau^2_n \mathbf{A}(\phivec) \mathbf{X}_{[-j]}\mathbf{X}_{[-j]}^{T} \mathbf{A}^{T}(\phivec) +I]^{-1} \mathbf{A}(\phivec)\\
    &=\mathbf{A}^{T}(\phivec) \left[I - (\mathbf{A}(\phivec) \mathbf{X}_{[-j]}) \left( \tau^{-2}_{n} I + 
    (\mathbf{A}(\phivec) \mathbf{X}_{[-j]})^{T} (\mathbf{A}(\phivec) \mathbf{X}_{[-j]})
    \right)^{-1} (\mathbf{A}(\phivec) \mathbf{X}_{[-j]})^{T}\right] \mathbf{A}(\phivec)\\
    &=\mathbf{A}^{T}(\phivec) \left[I - (\mathbf{X}_{[-j]}(\phivec)) \left( \tau^{-2}_{n} I + 
    (\mathbf{X}_{[-j]}(\phivec))^{T} (\mathbf{X}_{[-j]}(\phivec))
    \right)^{-1} (\mathbf{X}_{[-j]}(\phivec))^{T}\right] \mathbf{A}(\phivec)\\
\end{align*}
The second last equality comes from the identity $\left(I + UV \right)^{-1} = I - U \left(I + VU \right)^{-1} V$ and last inequality comes using the definition of $\mathbf{X}_n(\phivec)$.
If we define \[\mathbf{H}_{j}:=(\mathbf{X}_{[-j]}(\phivec)) \left( \tau^{-2}_{n} I + 
    (\mathbf{X}_{[-j]}(\phivec))^{T} (\mathbf{X}_{[-j]}(\phivec))
    \right)^{-1} (\mathbf{X}_{[-j]}(\phivec))^{T},\]
    we can write;
\[
\mathbf{\tilde{H}}_j = I- \mathbf{A}^{T}(\phivec) \left[I - \mathbf{H}_{j}\right] \mathbf{A}(\phivec)
\]
Now spectral value decomposition provides $\frac{\mathbf{X}_{[-j]}(\phivec)}{\sqrt{n}}= \mathbf{S}\mathbf{V} \mathbf{D}$.
\begin{align*}
I - \mathbf{H}_j &= I - \mathbf{X}_{[-j]}(\phivec)(\mathbf{X}_{[-j]}(\phivec)^T \mathbf{X}_{[-j]}(\phivec) + \tau^{-2}_n I)^{-1}\mathbf{X}_{[-j]}(\phivec)^T \\
              &= I - \mathbf{S}\mathbf{V}\mathbf{D}(\mathbf{D}^T \mathbf{V}^T \mathbf{V} \mathbf{D} + n^{-1}\tau^{-2}_n \mathbf{D}^T \mathbf{D})^{-1}\mathbf{D}^T \mathbf{V}^T \mathbf{S}^T \\
              &= I - \mathbf{S}\mathbf{V}(\mathbf{V}^T \mathbf{V} + n^{-1}\tau^{-2}_n I)^{-1} \mathbf{V}^T \mathbf{S}^T \\
              &= I - \mathbf{S}\mathbf{V} \ \text{diag} \left( \frac{1}{\lambda_1^2 + n^{-1}\tau^{-2}_n}, \ldots, \frac{1}{\lambda_n^2 + n^{-1}\tau^{-2}_n}, \frac{1}{n^{-1}\tau^{-2}_n}, \ldots, \frac{1}{n^{-1}\tau^{-2}_n} \right) \mathbf{V}^T \mathbf{S}^T \\
              &= \mathbf{S}\mathbf{S}^T - \mathbf{S} \ \text{diag} \left( \frac{\lambda_1^2}{\lambda_1^2 + n^{-1}\tau^{-2}_n}, \ldots, \frac{\lambda_n^2}{\lambda_n^2 + n^{-1}\tau^{-2}_n} \right) \mathbf{S}^T\\
              &= \mathbf{S} \ \text{diag} \left( \frac{n^{-1}\tau^{-2}_n}{\lambda_1^2 + n^{-1}\tau^{-2}_n}, \ldots, \frac{n^{-1}\tau^{-2}_n}{\lambda_n^2 + n^{-1}\tau^{-2}_n} \right) \mathbf{S}^T
\end{align*}
So, 
\[
(I - \mathbf{\tilde{H}}_j)= \mathbf{A}^{T}(\phivec)\mathbf{S} \ \text{diag} \left( \frac{n^{-1}\tau^{-2}_n}{\lambda_1^2 + n^{-1}\tau^{-2}_n}, \ldots, \frac{n^{-1}\tau^{-2}_n}{\lambda_n^2 + n^{-1}\tau^{-2}_n} \right) \mathbf{S}^T \mathbf{A}(\phivec)
\]
Now by Lemma \ref{lem:bound_by_eigen_value}, that if $\lambda_{1}$ is the maximum and $\lambda_{n}$ is the minimum non negative singular values ,

\[
\frac{n^{-1}\tau^{-2}_n}{\lambda_1^2 + n^{-1}\tau^{-2}_n}  \lambda_{min}\left(\mathbf{A}^{T}(\phivec) \mathbf{A}(\phivec)\right)\mathbf{I}   \underset{\sim}{\leq} 
(I - \mathbf{\tilde{H}}_j) \underset{\sim}{\leq} \frac{n^{-1}\tau^{-2}_n}{\lambda_n^2 + n^{-1}\tau^{-2}_n}  \lambda_{max}\left(\mathbf{A}^{T}(\phivec) \mathbf{A}(\phivec)\right)\mathbf{I}
\]
where $\mathbf{P}\underset{\sim}{\leq} \mathbf{Q} $ implies $\mathbf{Q}-\mathbf{P}$ is positive semi definite matrix. Now as $\mathbf{A}(\phivec)$ is a square matrix (specifically lower triangular), let $\mathbf{A}^{T}(\phivec)\mathbf{A}(\phivec)$ has the eigenvalues $d_{1}^{2},d_{2}^{2},\cdots,d_{n}^{2}$. Let $d_{\text{max}}^{2}=\text{max}\{d_{1}^{2},d_{2}^{2},\cdots,d_{n}^{2}\}$ and $d_{\text{min}}^{2}=\text{min}\{d_{1}^{2},d_{2}^{2},\cdots,d_{n}^{2}\}$. So, if $\frac{n^{-1}\tau^{-2}_n}{\lambda_1^2 + n^{-1}\tau^{-2}_n} d_{\text{min}}^{2} $ and $\frac{n^{-1}\tau^{-2}_n}{\lambda_n^2 + n^{-1}\tau^{-2}_n} d_{\text{max}}^{2}$ are in between 0 and 1 we should reach our condition. In general, if $d_{\text{max}}^{2} \leq 1$ we can reach our condition. 
\end{proof}

\noindent {\bf Proof of Lemma \ref{lem:bound_by_eigen_value}}:

\begin{proof}
    For \( u \neq 0 \),
\[
\frac{u^T X^T A X u}{u^T u} = \frac{u^T X^T A X u}{u^T X^T X u} \cdot \frac{u^T X^T X u}{u^T u}.
\]
Thus,
\[
\lambda^*_{\max}(X^T A X) = \max_{u \neq 0, Xu \neq 0} \frac{u^T X^T A X u}{u^T u}
\]
can be rewritten as
\[
\lambda^*_{\max}(X^T A X) = \max_{u \neq 0, Xu \neq 0} \frac{u^T X^T A X u}{u^T X^T X u} \cdot \frac{u^T X^T X u}{u^T u}.
\]
By Lemma \ref{lem:basic_inequality}, we know
\begin{align*}
&\lambda^*_{\max}(X^T A X) \\
&\leq \max_{u \neq 0, Xu \neq 0} \frac{u^T A u}{u^T u} \max_{u \neq 0, Xu \neq 0} \frac{u^T X^T X u}{u^T u} \\
&\leq \max_{v \neq 0} \frac{v^T A v}{v^T v}  \max_{u \neq 0, Xu \neq 0} \frac{u^T X^T X u}{u^T u} \\
&= \lambda_{\max}(A) \lambda_{\max}(X^T X), 
\end{align*}
So,
\[
\lambda^*_{\max}(X^T A X) \leq \lambda_{\max}(A) \lambda^*_{\max}(X^T X).
\]
Similarly, 
\[
\lambda^*_{\min}(X^T A X) = \max_{u \neq 0, Xu \neq 0} \frac{u^T X^T A X u}{u^T u}
\]
can be rewritten as
\[
\lambda^*_{\min}(X^T A X) = \max_{u \neq 0, Xu \neq 0} \frac{u^T X^T A X u}{u^T X^T X u} \cdot \frac{u^T X^T X u}{u^T u}.
\]
Again by Lemma \ref{lem:basic_inequality}, we know
\begin{align*}
&\lambda^*_{\min}(X^T A X) \\
&\geq \min_{u \neq 0, Xu \neq 0} \frac{u^T A u}{u^T u} \min_{u \neq 0, Xu \neq 0} \frac{u^T X^T X u}{u^T u} \\
&\geq \min_{v \neq 0} \frac{v^T A v}{v^T v}  \min_{u \neq 0, Xu \neq 0} \frac{u^T X^T X u}{u^T u} \\
&= \lambda_{\min}(A) \lambda_{\min}(X^T X), 
\end{align*}
So,
\[
\lambda^*_{\min}(X^T A X) \geq \lambda_{\min}(A) \lambda^*_{\min}(X^T X).
\]
\end{proof}

Proof of Lemma \ref{lem:basic_inequality} is straightforward. Now we show the proof of Lemma \ref{lem:bounded_eigenvalue}.
\begin{proof}
Observe that, 
    \[
\lambda_{\text{min}} \mathbf{x}^T \mathbf{x} \leq \mathbf{x}^T (I - \mathbf{\tilde{H}}_j ) \mathbf{x} \leq \lambda_{\text{max}} \mathbf{x}^T \mathbf{x}
\]
where $\lambda_{\text{min}}$ and $\lambda_{\text{max}}$ are minimum and maximum eigen values of $(I - \mathbf{\tilde{H}}_j )$ respectively, which are bounded.

\end{proof}

\noindent {\bf Proof of Lemma \ref{lem:order_for_different_index}}:

\begin{proof}
    Now observe that,
\begin{align*}
    &(I - \mathbf{\tilde{H}}_j )\mathbf{X}_{[-j]}\\
    &=\mathbf{A}^{T}(\phivec) \left[I - \mathbf{H}_{j}\right] \mathbf{A}(\phivec)\mathbf{X}_{[-j]}\\
    &=\mathbf{A}^{T}(\phivec)\mathbf{A}(\phivec)\mathbf{X}_{[-j]}-\mathbf{A}^{T}(\phivec) \mathbf{H}_{j}\mathbf{A}(\phivec)\mathbf{X}_{[-j]}\\
    &=\mathbf{A}^{T}(\phivec) \left[I - \mathbf{H}_{j}\right] \mathbf{A}(\phivec)\mathbf{X}_{[-j]}\\
    &=\mathbf{A}^{T}(\phivec)\mathbf{A}(\phivec)\mathbf{X}_{[-j]}-\mathbf{A}^{T}(\phivec) (\mathbf{X}_{[-j]}(\phivec)) \left( \tau^{-2}_{n} I + 
    (\mathbf{X}_{[-j]}(\phivec))^{T} (\mathbf{X}_{[-j]}(\phivec))
    \right)^{-1} (\mathbf{X}_{[-j]}(\phivec))^{T}\mathbf{A}(\phivec)\mathbf{X}_{[-j]}\\
    &=\mathbf{A}^{T}(\phivec)\mathbf{X}_{[-j]}(\phivec)- \mathbf{A}^{T}(\phivec) (\mathbf{X}_{[-j]}(\phivec)) \left( \tau^{-2}_{n} I + 
    (\mathbf{X}_{[-j]}(\phivec))^{T} (\mathbf{X}_{[-j]}(\phivec))
    \right)^{-1} (\mathbf{X}_{[-j]}(\phivec))^{T}\mathbf{X}_{[-j]}(\phivec)\\
    &=\mathbf{A}^{T}(\phivec)\mathbf{X}_{[-j]}(\phivec)- \\
    &\mathbf{A}^{T}(\phivec) (\mathbf{X}_{[-j]}(\phivec)) \left( \tau^{-2}_{n} I + 
    (\mathbf{X}_{[-j]}(\phivec))^{T} (\mathbf{X}_{[-j]}(\phivec))
    \right)^{-1} ((\mathbf{X}_{[-j]}(\phivec))^{T}\mathbf{X}_{[-j]}(\phivec)+\tau^{-2}_{n} I-\tau^{-2}_{n} I)\\
    &=\mathbf{A}^{T}(\phivec)\mathbf{X}_{[-j]}(\phivec)- \mathbf{A}^{T}(\phivec)\mathbf{X}_{[-j]}(\phivec)+\tau^{-2}_{n}\mathbf{A}^{T}(\phivec) (\mathbf{X}_{[-j]}(\phivec)) \left( \tau^{-2}_{n} I + 
    (\mathbf{X}_{[-j]}(\phivec))^{T} (\mathbf{X}_{[-j]}(\phivec))
    \right)^{-1}\\
    &=\tau^{-2}_{n}\mathbf{A}^{T}(\phivec) (\mathbf{X}_{[-j]}(\phivec)) \left( \tau^{-2}_{n} I + 
    (\mathbf{X}_{[-j]}(\phivec))^{T} (\mathbf{X}_{[-j]}(\phivec))
    \right)^{-1}\\
    &=\mathbf{A}^{T}(\phivec) (\mathbf{X}_{[-j]}(\phivec)) \left( I + \tau^{2}_{n} 
    (\mathbf{X}_{[-j]}(\phivec))^{T} (\mathbf{X}_{[-j]}(\phivec))
    \right)^{-1}\\
\end{align*}
Now,
\[ \mathbf{x}_{j}^T (I - \mathbf{\tilde{H}}_j )\mathbf{X}_{[-j]}=\mathbf{x}_{j}^T \mathbf{A}^{T}(\phivec) (\mathbf{X}_{[-j]}(\phivec)) \left( I + \tau^{2}_{n} 
    (\mathbf{X}_{[-j]}(\phivec))^{T} (\mathbf{X}_{[-j]}(\phivec))
    \right)^{-1}:=\mathbf{\omega}.\]
So we can write \[ \mathbf{x}_{j}^T\mathbf{A}^{T}(\phivec) (\mathbf{X}_{[-j]}(\phivec)) =\mathbf{\omega} \left( I + \tau^{2}_{n} 
    (\mathbf{X}_{[-j]}(\phivec))^{T} (\mathbf{X}_{[-j]}(\phivec))
    \right) \]
Now $(\mathbf{A}(\phivec)\mathbf{x}_{j})^{T}(\mathbf{X}_{[k,-j]}(\phivec))=n\rho_{jk}$ for any $j \in S^{*c}$ and $k \in S_{\beta}$, we can obtain that $\omega_{k}/\sqrt{n}<\infty$ following similar steps in page 390 of \cite{chen2019fast}.
\end{proof}

Now we show the proof of Lemma \ref{lem:l2_convergence}.
\begin{proof}
    From Lemma \ref{lem:bounded_eigenvalue}, we know that $(I - \mathbf{\tilde{H}}_j )$ have bounded eigen values and Thus $(I - \mathbf{\tilde{H}}_j )^{2}$ also do have the same. Now from the expression \ref{mu} and the conditions given, we have
    \[
\mathbb{E}(\mu_{1j}) \to \beta^*_j, \quad \text{Var}(\mu_{1j}) \to 0, \quad \mathbb{E}(\mu_{0j}) \to 0, \quad \text{and} \quad \text{Var}(\mu_{0j}) \to 0.
\]
\end{proof}
Now,
\begin{align*}
    &\Pr(Z^{\beta}_{j} = 1 \mid Y)= 1 - \Pr(Z^{\beta}_{j} = 0 \mid Y)= \frac{q^{\beta}_{jn} F_{1j}}{q^{\beta}_{jn} F_{1j} + (1 - q^{\beta}_{jn}) F_{0j}}=\frac{q^{\beta}_{jn} \frac{F_{1j}}{F_{0j}}}{q^{\beta}_{jn} \frac{F_{1j}}{F_{0j}} + (1 - q^{\beta}_{jn}) }\\
\end{align*}
\begin{align}
    \frac{F_{1j}}{F_{0j}}
    &=\frac{\sqrt{\frac{(\tau^{\beta}_{1n})^{-2}}{\mathbf{x}_j^T (I - \mathbf{\tilde{H}}_j) \mathbf{x}_j + (\tau^{\beta}_{1n})^{-2}}} \times \left( b + \frac{\mathbf{y}^T (I - \mathbf{\tilde{H}}_j) \mathbf{y}}{2} - \frac{(\mathbf{x}_j^T (I - \mathbf{\tilde{H}}_j) \mathbf{y})^2}{2 (\mathbf{x}_j^T (I - \mathbf{\tilde{H}}_j) \mathbf{x}_j + (\tau^{\beta}_{1n})^{-2})} \right)^{-\left( \frac{n}{2} + a \right)}}
{\sqrt{\frac{(\tau^{\beta}_{0n})^{-2}}{\mathbf{x}_j^T (I - \mathbf{\tilde{H}}_j) \mathbf{x}_j + (\tau^{\beta}_{0n})^{-2}}} \times \left( b + \frac{\mathbf{y}^T (I - \mathbf{\tilde{H}}_j) \mathbf{y}}{2} - \frac{(\mathbf{x}_j^T (I - \mathbf{\tilde{H}}_j) \mathbf{y})^2}{2 (\mathbf{x}_j^T (I - \mathbf{\tilde{H}}_j) \mathbf{x}_j + (\tau^{\beta}_{0n})^{-2})} \right)^{-\left( \frac{n}{2} + a \right)}}\\
&=\sqrt{\frac{\tau_{0n}^{2}\mathbf{x}_j^T (I - \mathbf{\tilde{H}}_j) \mathbf{x}_j + 1}{(\tau^{\beta}_{1n})^{2}\mathbf{x}_j^T (I - \mathbf{\tilde{H}}_j) \mathbf{x}_j + 1}} 
\frac{\left( b + \frac{\mathbf{y}^T (I - \mathbf{\tilde{H}}_j) \mathbf{y}}{2} - \frac{(\mathbf{x}_j^T (I - \mathbf{\tilde{H}}_j) \mathbf{y})^2}{2 (\mathbf{x}_j^T (I - \mathbf{\tilde{H}}_j) \mathbf{x}_j + (\tau^{\beta}_{1n})^{-2})} \right)^{-\left( \frac{n}{2} + a \right)}}{\left( b + \frac{\mathbf{y}^T (I - \mathbf{\tilde{H}}_j) \mathbf{y}}{2} - \frac{(\mathbf{x}_j^T (I - \mathbf{\tilde{H}}_j) \mathbf{y})^2}{2 (\mathbf{x}_j^T (I - \mathbf{\tilde{H}}_j) \mathbf{x}_j + (\tau^{\beta}_{0n})^{-2})} \right)^{-\left( \frac{n}{2} + a \right)}}\\
&=\sqrt{\frac{\tau_{0n}^{2}\mathbf{x}_j^T (I - \mathbf{\tilde{H}}_j) \mathbf{x}_j + 1}{(\tau^{\beta}_{1n})^{2}\mathbf{x}_j^T (I - \mathbf{\tilde{H}}_j) \mathbf{x}_j + 1}} 
\frac{\left( 2b + \mathbf{y}^T (I - \mathbf{\tilde{H}}_j) \mathbf{y} - (\mathbf{x}_j^T (I - \mathbf{\tilde{H}}_j) \mathbf{x}_j + (\tau^{\beta}_{0n})^{-2}) \mu_{0j}^{2} 
\right)^{\left( \frac{n}{2} + a \right)}}{\left( 2b + \mathbf{y}^T (I - \mathbf{\tilde{H}}_j) \mathbf{y} - (\mathbf{x}_j^T (I - \mathbf{\tilde{H}}_j) \mathbf{x}_j + (\tau^{\beta}_{1n})^{-2}) \mu_{1j}^{2} \right)^{\left( \frac{n}{2} + a 
\right)}}
\end{align}
Now using the similar strategy in Lemma 5 from \cite{chen2019fast}, we can show that when $n\tau^{\beta}_{0n}\rightarrow 0$ as $n\rightarrow \infty$, $(\mathbf{x}_j^T (I - \mathbf{\tilde{H}}_j) \mathbf{x}_j + (\tau^{\beta}_{0n})^{-2}) \mu_{0j}^{2}\rightarrow 0.$
If we do follow the similar conditions mentioned in \cite{chen2019fast} in page 292. \newline
Recall the conditions in Condition~\ref{cond:1}, \ref{cond:2}, and \ref{cond:3}. Based on this, the proofs of Theorems~\ref{theorem:1} and \ref{theorem:2} follow a similar approach to \cite{chen2019fast}, specifically on pages 293 and 295. With the additional conditions \ref{cond:2_phi} and \ref{cond:3_phi}, we can establish both individual and pairwise consistency for $\phivec$ in Theorems~\ref{theorem:1_phi} and \ref{theorem:2_phi}. These proofs also closely resemble those on page 293 and in \cite{chen2019fast} on page 295.

\section{Appendix: Additional analysis of the real data}
\label{appendixB}
For our analysis, we considered two datasets. The first dataset focuses on water table depth prediction. Water table depth and elevation are two correlated variables. So we did not include the elevation as a predictor. For more details please consider \cite{manna2024hydrology}. Time series plots for various weather variables are shown in Figure~\ref{fig:wtd_bvariables}. The complex correlations among different variables and their lagged values are illustrated in Figure~\ref{fig:correlation_wtd}. The partial autocorrelation for different weather variables is demonstrated in Figure~\ref{fig:pacf_wtd}. The plot for inclusion probability of the variables are given in Figure~\ref{fig:inc_wtd_vars} and Figure~~\ref{fig:inc_wtd_lags}. We can observe that up to 5 lags were sufficient in the error process for the water table depth estimation.

For the second dataset, which is used for stock market analysis and S\&P 500 prediction, we considered the following stocks and their four lag values in the variable set:
\begin{flushleft}
\texttt{\^{}GSPC, AAPL, MSFT, GOOGL, AMZN, TSLA, META, NVDA, BRK-B, JNJ, V, MA,}\\
\texttt{DIS, HD, NFLX, PFE, ADBE, INTC, CMCSA, CSCO, KO, PEP, T, MRK, WMT, ABT,}\\
\texttt{XOM, CVX, MCD, NKE, UNH, GS, JPM, MS, BA, IBM, BMY, ORCL, AMGN, GILD,}\\
\texttt{TXN, CVS, COST, LMT, EA, WFC}
\end{flushleft}
GSPC is considered as S\&P 500 ticker in our analysis. The selected set of stocks represents a well-diversified subset of the S\&P 500 index. It includes major constituents across key sectors such as Technology (e.g., AAPL, MSFT, NVDA), Healthcare (e.g., JNJ, UNH, PFE), Financials (e.g., JPM, BRK-B, GS), Consumer Discretionary and Staples (e.g., AMZN, WMT, KO), Energy (e.g., XOM, CVX), and Industrials (e.g., BA, LMT). This diversity ensures that the dataset captures broad market dynamics and sectoral behavior, making it suitable for representative analysis of the S\&P 500 performance.
The correlation structure among these stocks is also complex. Figure~\ref{fig:log_return} demonstrate the original log return of the stocks that we incorporated for our analysis. Log returns themselves are not directly lag-related across stocks. Each individual stock is also correlated with itself demonstrated in the Figure~\ref{fig:pacf_snp}. The plots of inclusion probability for different stocks in the initial set is given in the Figure~\ref{fig:inc_snp_vars}. We observe that the lags in the error process is almost zero in the stock market data.
\begin{figure}
    \centering
    \includegraphics[width=\linewidth]{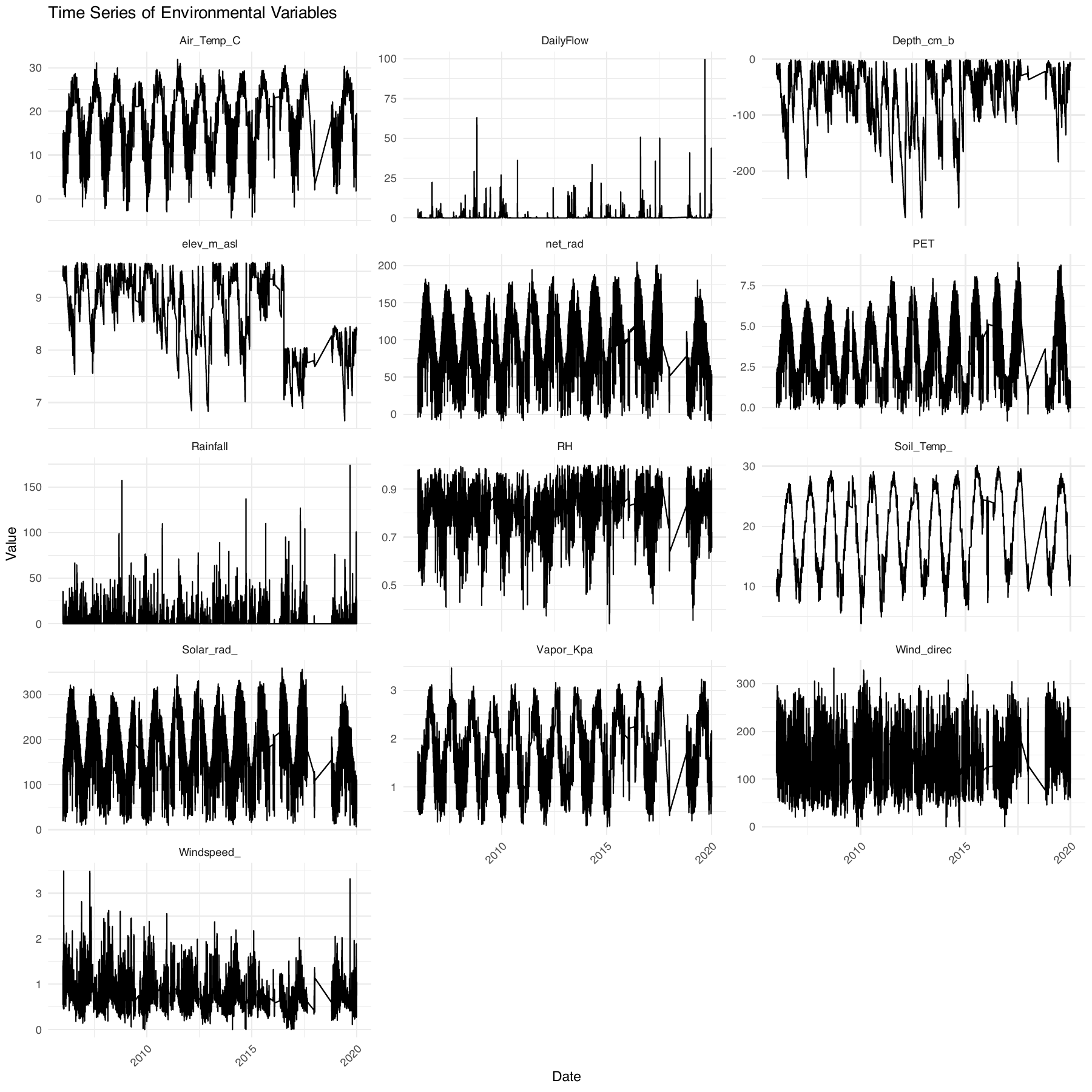}
    \caption{Different weather variables daily time series plot used for water table depth prediction}
    \label{fig:wtd_bvariables}
\end{figure}

\begin{figure}
    \centering
    \includegraphics[width=1\linewidth]{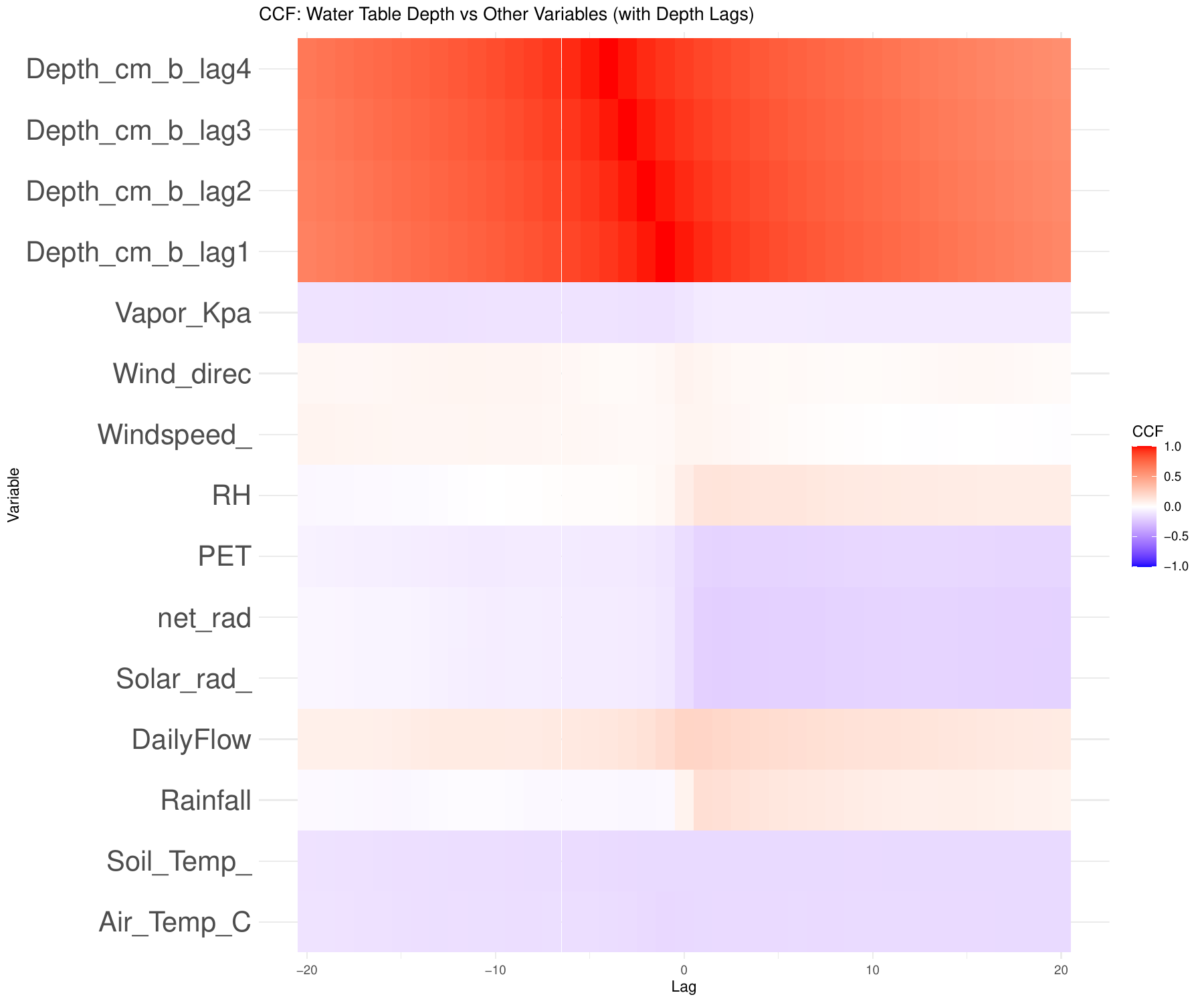}
    \caption{CCF for water table depth variables with the other weather variables across different lags}
    \label{fig:correlation_wtd}
\end{figure}

\begin{figure}
    \centering
    \includegraphics[width=1\linewidth]{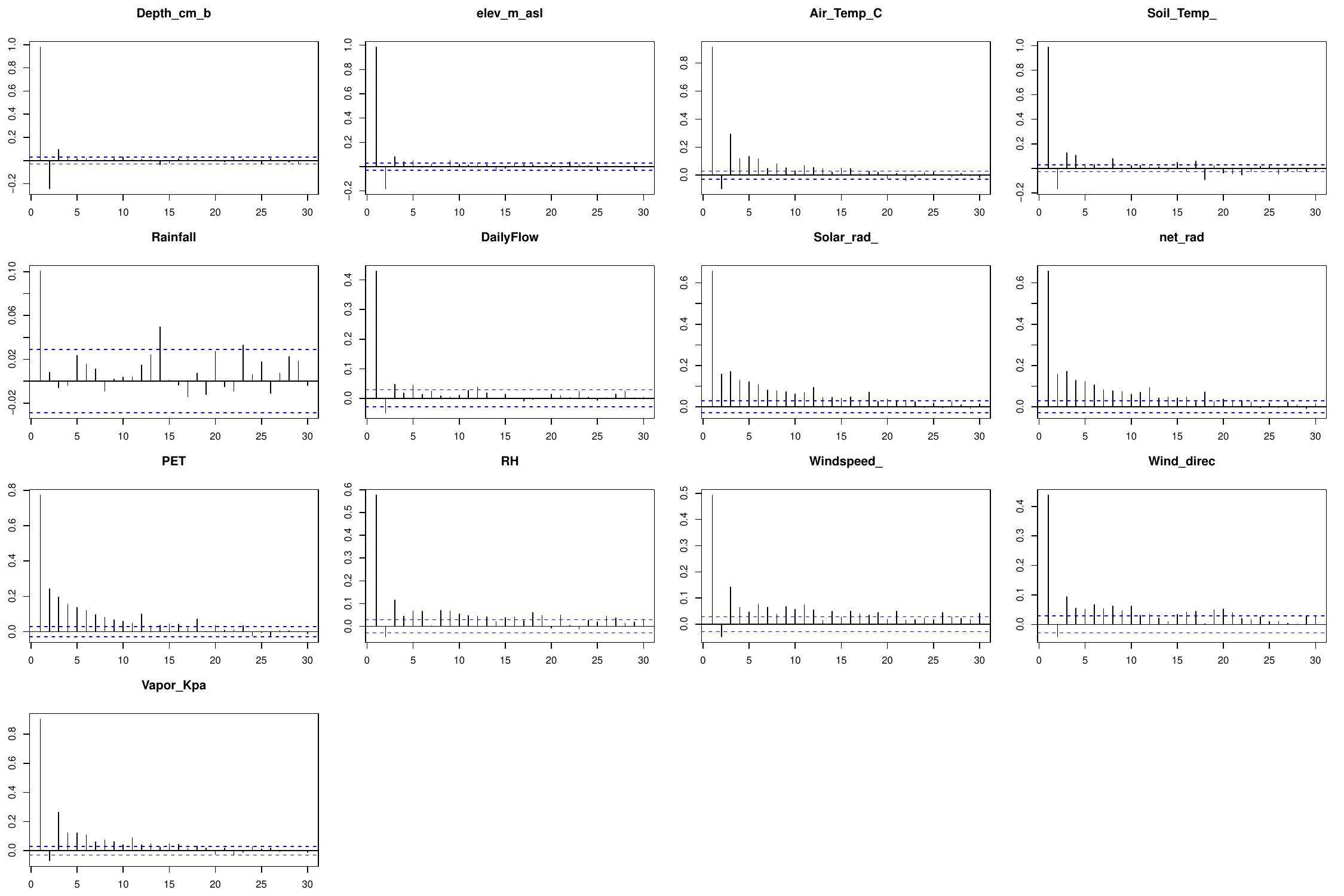}
    \caption{PACF for different weather variables}
    \label{fig:pacf_wtd}
\end{figure}


\begin{figure}
    \centering
    \includegraphics[width=1\linewidth]{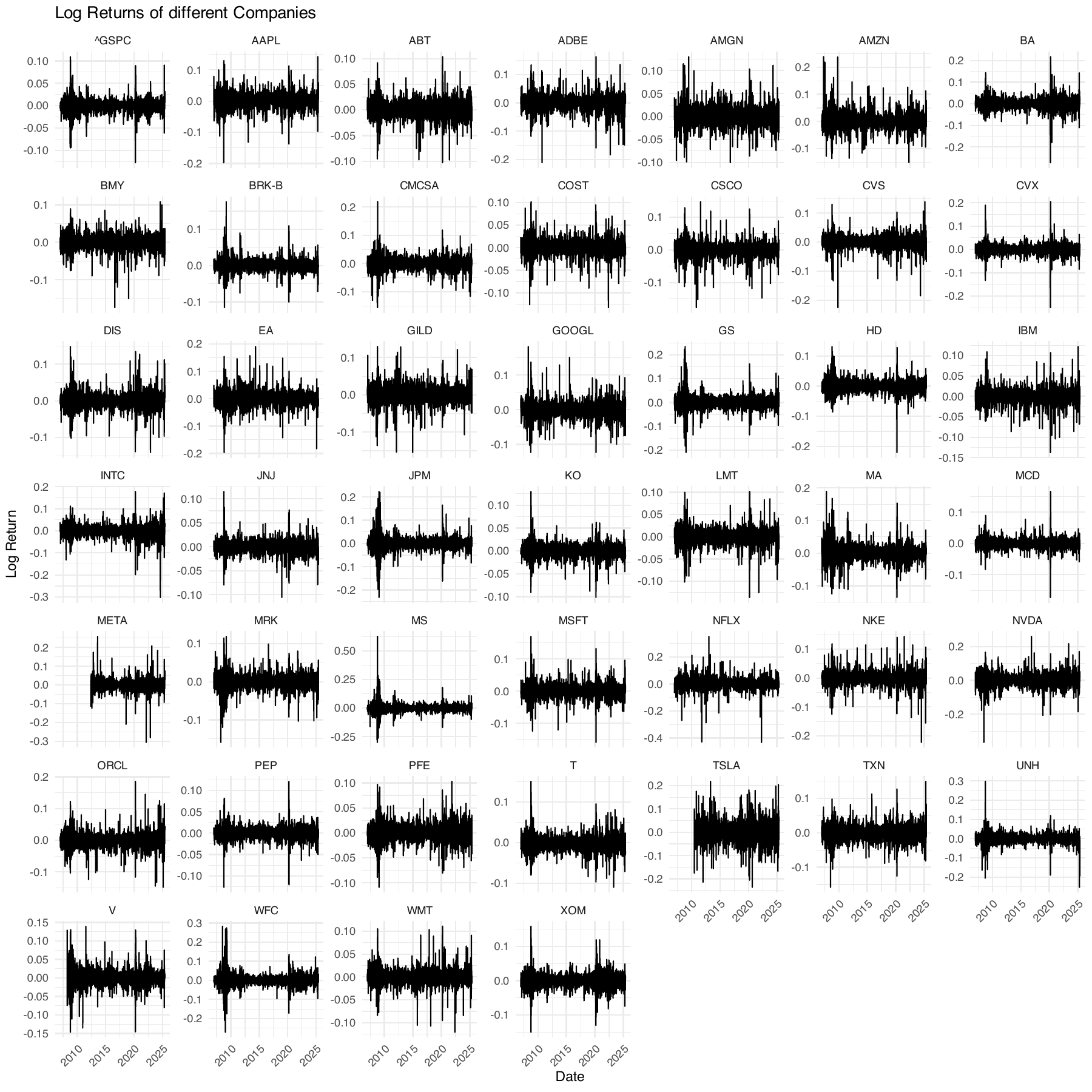}
    \caption{Log return for different stocks}
    \label{fig:log_return}
\end{figure}

\begin{figure}
    \centering
    \includegraphics[width=1\linewidth]{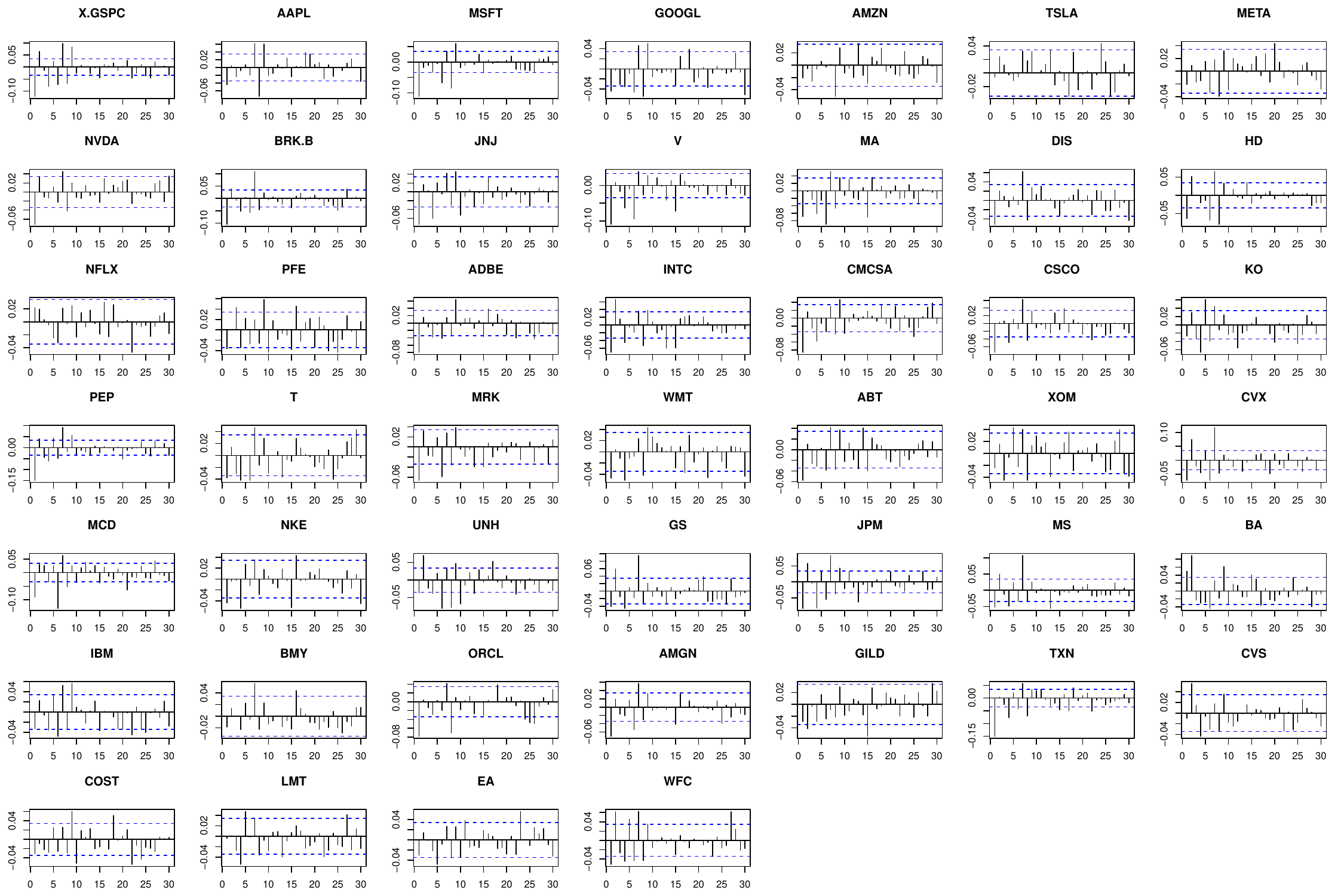}
    \caption{PACF for different stocks}
    \label{fig:pacf_snp}
\end{figure}

\begin{figure}
    \centering
    \includegraphics[width=.5\linewidth]{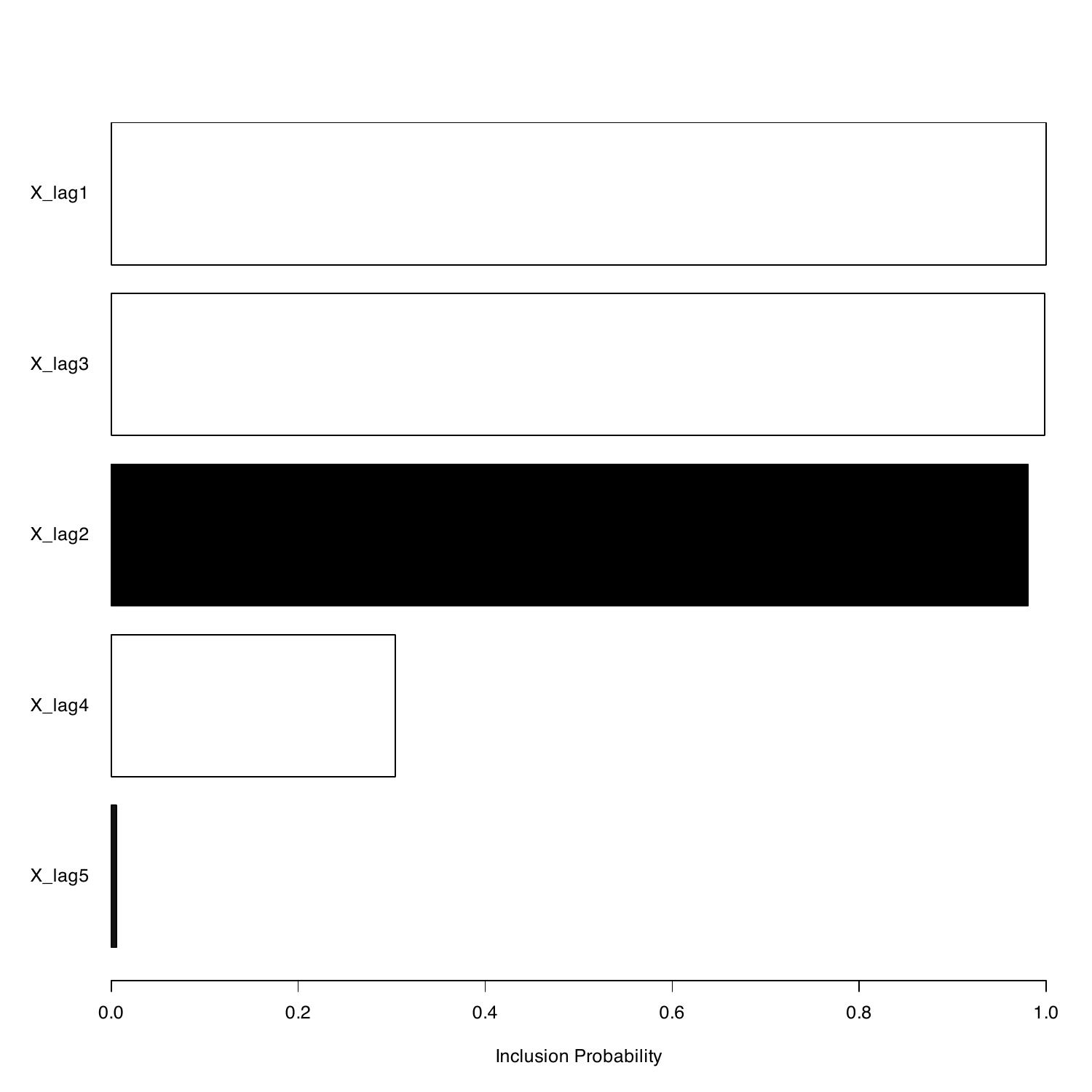}
    \caption{Inclusion probability for different lags of the process while predicting the water table depth using $\frac{2}{3}$ of the data}
    \label{fig:inc_wtd_lags}
\end{figure}

\begin{figure}
    \centering
    \includegraphics[width=.5\linewidth]{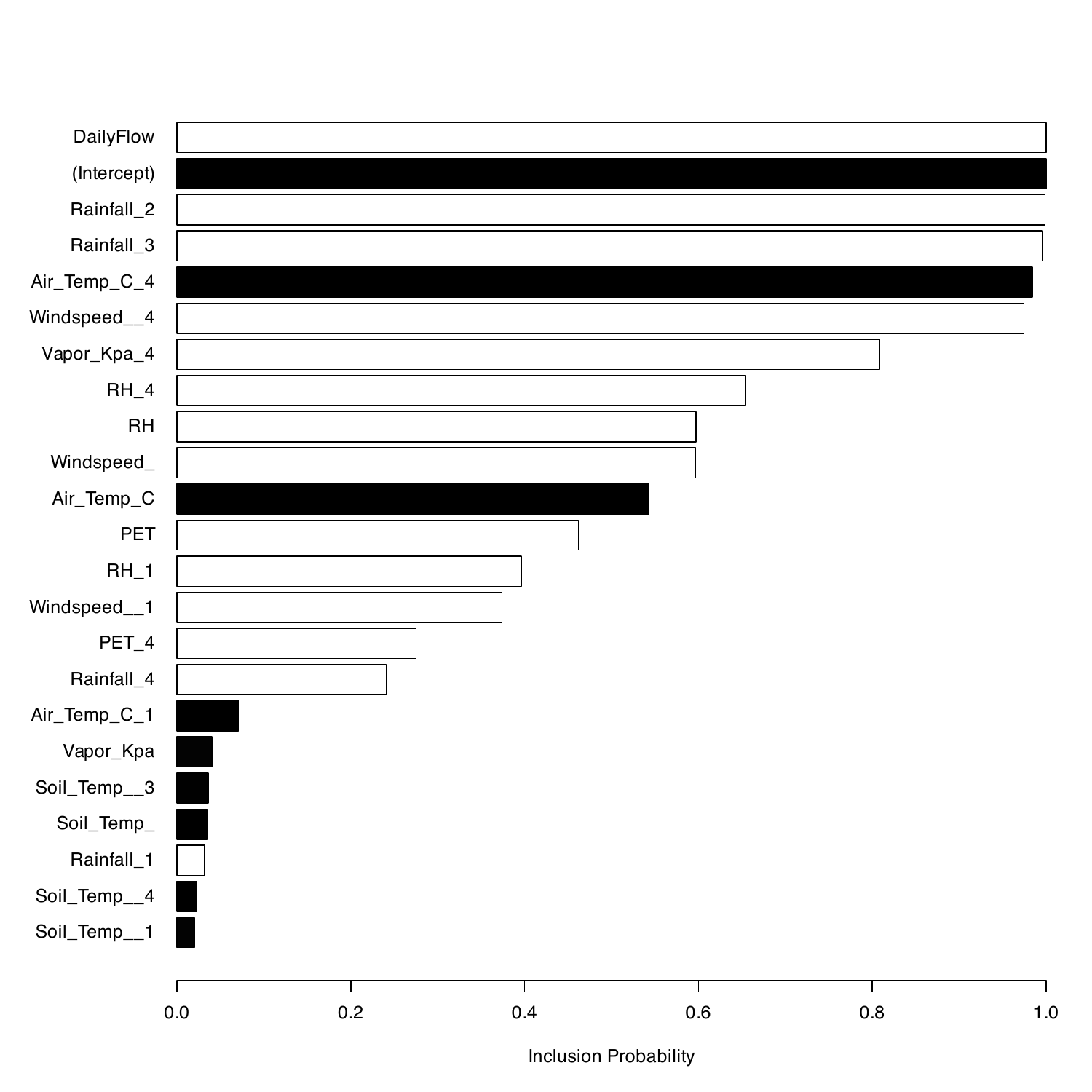}
    \caption{Inclusion probability for different covariates used for prediction of water table depth using $\frac{2}{3}$ of the data}
    \label{fig:inc_wtd_vars}
\end{figure}

\begin{figure}
    \centering
    \includegraphics[width=.5\linewidth]{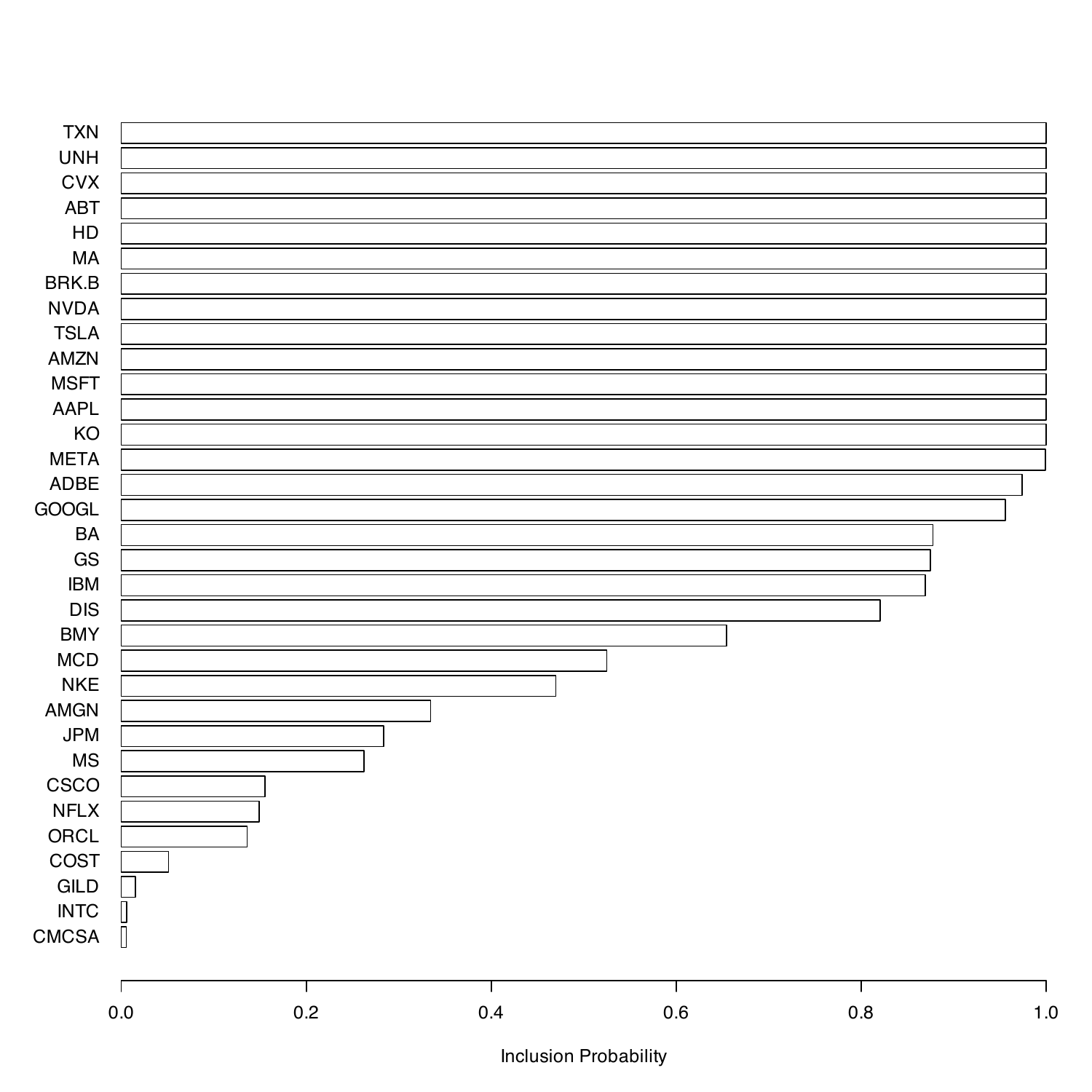}
    \caption{Inclusion probability for different stocks used for prediction of S\&P 500 using $\frac{2}{3}$ of the data}
    \label{fig:inc_snp_vars}
\end{figure}


\end{document}